%% file: main.tex
\documentclass[envcountsame]{llncs}
           
\let\llncssubparagraph\subparagraph
\let\subparagraph\paragraph
\let\subparagraph\llncssubparagraph
\usepackage{microtype}

\usepackage{hyperref}
\usepackage{graphicx}
\usepackage{amsmath}
\usepackage{amsthm}

\usepackage{amsfonts}
\usepackage{amstext}
\usepackage{amssymb}
\usepackage{amsthm}
\usepackage{marvosym}

\usepackage{stmaryrd}
\usepackage{mathrsfs}
\usepackage{calc}
\usepackage{xspace}
\usepackage{mathtools}
\usepackage{macros}
\usepackage{tikz}
\usepackage[capitalise]{cleveref}
\usepackage{thmtools}
\usepackage{thm-restate}
\usepackage{subcaption}
\include{tikzsetup}
\newcommand{\mypar}[1]{\subsubsection*{#1}}

\makeatletter
\def\orcidID#1{\smash{\href{http://orcid.org/#1}{\protect\raisebox{-1.25pt}{
\protect\includegraphics{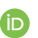}}}}}
\makeatother

\begin{document}
\title{Simple Stochastic Games with Almost-Sure Energy-Parity Objectives are in NP\ and coNP}
\author{Richard Mayr \inst{1}
\and Sven Schewe \inst{2}\orcidID{0000-0002-9093-9518}
\and Patrick Totzke \inst{2}\orcidID{0000-0001-5274-8190}
\and \\ Dominik Wojtczak \inst{2\text{(\Letter)}}\orcidID{0000-0001-5560-0546}}
\authorrunning{R.~Mayr, S.~Schewe, P.~Totzke, D.~Wojtczak}
\institute{University of Edinburgh, Edinburgh, UK\\
\and University of Liverpool, Liverpool, UK\\
}
\maketitle              %
\begin{abstract}
    \input{abstract}

\end{abstract}

\section{Introduction}
\input{sec.intro}

\renewcommand{\O}{\mathsf{Obj}}
\section{Preliminaries}\label{sec:preliminaries}
\input{sec.prelim}

\section{Characterizing Energy-Parity via Gain and Bailout}
\label{sec:correctness}
\input{sec.characterization}

\input{proof.correctness}

\section{Bailout}
\label{sec:bailout}
\input{sec.bailout}

\section{Gain}
\label{sec:gain}
\input{sec.gain}

\section{The Main Results}
\label{sec:main}
\input{sec.mainresults}

\input{sec.storage-parity}

\section{Conclusion and Outlook}
\label{sec:conclusion}
\input{sec.conclusion}

\section*{Acknowledgments}

The work of Sven Schewe and Dominik Wojtczak was supported by EPSRC grant EP/P020909/1.

\newpage
\bibliographystyle{splncs03}
\bibliography{conferences, journals, autocleaned}

\newpage
\appendix

\section{Lifting Almost-sure Strategies from MDPs to SSGs}
\label{app:lifting}
\input{app.lifting}

\newpage
\section{Bailout} %
\label{sec:appendixC}
\label{app:bailout}
\input{app.4}

\newpage
\section{Gain} %
\label{app:gain}
\input{app.gain}
\end{document}

%% file: tikzsetup.tex
\usetikzlibrary{positioning}
\usetikzlibrary{arrows}
\usetikzlibrary{petri}
\tikzset{>=latex}
\tikzset{cstate/.style={rectangle,draw,inner sep=2mm, }}%
\tikzstyle{every transition}=[rectangle,draw=black!50,fill=black!20,thick,inner sep=2mm,rounded corners=3pt]

%% file: abstract.tex
We study stochastic games with energy-parity objectives, which combine quantitative rewards with a qualitative $\omega$-regular condition:
The maximizer aims to avoid running out of energy while simultaneously satisfying a parity condition.
We show that the corresponding almost-sure problem, i.e., checking whether
there exists a maximizer strategy
that achieves the energy-parity objective with probability $1$ when starting at a given energy level $k$,
is decidable and in $\NP \cap \coNP$. 
The same holds for checking if such a $k$ exists and if a given $k$ is minimal.

%% file: sec.intro.tex
\emph{Simple stochastic games} (SSGs), also called \emph{competitive Markov decision
    processes} \cite{Filar_Vrieze:book},
or \emph{$2\frac{1}{2}$-player games} \cite{CJH2004,CJH2003}
are turn-based games of perfect information played on finite graphs.
Each state is either random or belongs to one of the players (maximizer or minimizer).
A game is played successively moving a pebble along the game graph, where the next state is chosen by the player who owns the current one
or, in the case of random states, according to a predefined distribution.
This way, an infinite run is produced.
The maximizer tries to achieve an objective (in our case almost surely), while the minimizer tries to prevent this.
The maximizer can be seen as a controller trying to ensure an objective in the
face of both known random failure modes (encoded by the random states) and
an unknown or hostile environment (encoded by the minimizer player).

Stochastic games were first introduced in Shapley's seminal work \cite{shapley1953stochastic} in 1953 and have since then played a central role in the solution of many problems in
computer science, including synthesis of reactive systems \cite{ramadge1987supervisory,pnueli1989synthesis}; checking interface compatibility \cite{de2001interface};
well-formedness of specifications \cite{dill1989trace}; verification
of open systems \cite{alur2002alternating}; 
and many others.

A huge variety of objectives for such games was already studied in the literature. We will mainly focus on three of them in this paper: parity; mean-payoff; and energy objectives. In order to define them we assume that numeric rewards are assigned to transitions, and priorities (encoded by bounded
non-negative numbers) are assigned to states.

The {\em parity objective} simply
asks that the minimal priority that appears infinitely often in a run is even.
Such a condition is a canonical way to define desired behaviors of systems, such as safety, liveness,
fairness, etc.;
it subsumes all $\omega$-regular objectives.
The algorithmic problem of deciding the winner in non-stochastic parity games is polynomial-time equivalent to the model checking of
the modal $\mu$-calculus \cite{wilke2001alternating} and is at the center of the algorithmic solutions to the Church's synthesis problem \cite{rabin1972automata}.
But the impact of parity games goes well beyond automata theory and logic: They facilitated the solution of two long-standing open problems in stochastic planning \cite{fearnley2010exponential} and in linear
programming \cite{friedmann2011subexponential}, which was done by careful adaptation of the parity game examples on which the strategy improvement algorithm \cite{friedmann2009exponential} requires exponentially many iterations.

The parity objective can be seen as a special case of the {\em mean-payoff objective} that asks for the limit average reward per transition along the run to be non-negative. 
Mean-payoff objectives are among the first objectives studied for stochastic games and go back to a 1957 paper by Gillette \cite{gillette1957stochastic}.
They allow for reasoning about the efficiency of a system, e.g., how fast it operates once optimally controlled.

The {\em energy objective} \cite{chakrabarti2003resource} can be seen as a refinement of the mean-payoff objective. It asks for the accumulated reward at any point of a run not to be lower than some finite threshold.
As the name suggests, it is useful when reasoning about systems with a finite initial energy level that should never become depleted.
Note that the accumulated reward is not bounded a-priori, which essentially turns a finite-state game into an infinitely-state one.

In this paper we consider SSGs with \emph{energy-parity} objectives,
which requires runs to satisfy both an energy and a parity objective.
It is natural to consider such an objective for systems that should not only be correct, but also energy efficient.
For instance, consider a robot maintaining a nuclear power plant. We not only require the robot to correctly react to all possible chains of events (parity objective for functional correctness), but also never to run out of energy as charging it manually would be risky (energy objective).

While the complexity of games with single objectives is often in $\NP\cap \coNP$, asking for multiple objectives often makes solving games harder.
Parity games are commonly viewed as the simplest of these objectives, and some traditional solutions for non-stochastic games go through simple reductions to mean-payoff or energy conditions (which are quite similar in non-stochastic games) to discounted payoff games that establishes the membership of those problems in UP and coUP \cite{jurdzinski1998deciding}.
However,
asking for \emph{two} parity objectives to be satisfied at the same time
leads to $\coNP$ completeness \cite{DBLP:conf/fossacs/ChatterjeeHP07}.

We study the almost sure satisfaction of the energy-parity objective, i.e., with probability 1. Such {\em qualitative analysis} is important as there are many applications where we
need to know whether the correct behavior arises almost-surely, e.g.,
in the analysis of randomized distributed algorithms (see, e.g, \cite{pogosyants2000verification,stoelinga2003fun}) and safety-critical examples like the one from above.
Moreover, the algorithms for {\em quantitative analysis}, i.e., computing the optimal probability of satisfaction, typically start by 
performing the qualitative analysis first and then solving a game with a simpler objective (see, e.g., \cite{CJH2004,chatterjee2005complexity}). Finally, there are stochastic models for which qualitative analysis is decidable but quantitative one is not (e.g., probabilistic finite automata \cite{blondel2003undecidable}). This may also be the case for our model.  

\mypar{Our contributions.}
We consider stochastic games with energy-parity winning conditions and show that deciding whether
maximizer can win almost-surely for a given initial energy level $k$ 
is in $\NP \cap \coNP$. We show the same for checking if such $k$ exists at all and 
checking if a given $k$ is the smallest possible for which this holds.
The proofs are considerably harder than the corresponding result for MDPs
\cite{MSTW2017} (on which they are partly based), because the attainable mean-payoff value
is no longer a valid criterion in the analysis (via combinations of sub-objectives).
E.g., even though the stored energy
might be inexorably drifting towards $+\infty$ (resp.\ $-\infty$), the mean-payoff value
might still be zero because the minimizer (resp.\ maximizer) can
delay payoffs for longer and longer (though not indefinitely, due to the
parity condition). Moreover, the minimizer might be able to
choose between different ways of losing and never commit to any particular way
after any finite prefix of the play (see \cref{ex:two-ways}).

Our proof characterizes almost-sure energy-parity via a recursive combination of complex
sub-objectives called \emph{Gain} and \emph{Bailout}, which can each eventually be solved in
$\NP \cap \coNP$.

Our proof of the \coNP\ membership is based on a result 
on the strategy complexity of a natural class of objectives,
which is of independent interest.
We show 
(cf.~\cref{lem:FD-strategy}; based on previous work in \cite{GK2014})
that, if an objective $\?O$ is such that its complement %
is both shift-invariant and submixing,
and that every MDP admits optimal finite-memory
deterministic maximizer strategies for $\?O$, then the same is true in
turn-based stochastic games.

\input{ex.two-ways}

\mypar{Previous work on combined objectives.}
Non-stochastic energy-parity games have been studied in \cite{CD2010}. They
can be solved in $\NP \cap \coNP$
and maximizer strategies require only
finite (but exponential) memory, a property that also allowed to
show P-time inter-reducibility with mean-payoff parity games.
More recently they were also shown to be solvable in pseudo-quasi-polynomial time \cite{daviaud2018pseudo}.
Related results on non-stochastic games (e.g., mean-payoff parity) are
summarized in \cite{chatterjee2011games}. %

Most existing work
on combined objectives
for stochastic systems %
\cite{CD2011,chatterjee2011games,BKN2016,MSTW2017}
is restricted to
Markov decision processes (MDPs; aka $1\frac{1}{2}$-player games).
Almost-sure energy-parity objectives for MDPs were first considered in
\cite{CD2011,chatterjee2011games}, %
where a direct reduction to ordinary energy games was proposed.
This reduction relies on the assumption
that maximizer can win using finite memory if at all.
Unfortunately, this assumption does not necessarily hold:
it was shown in \cite{MSTW2017} that an almost sure winning strategy
for energy-parity in finite MDPs may require infinite memory.
Nevertheless, it was possible to recover the original result,
that deciding the existence of a.s.~winning strategies is
in $\NP \cap \coNP$ (and pseudo-polynomial time),
by showing that
the existence of an a.s.~winning strategy
can be witnessed by the existence of two compatible, and finite-memory, winning strategies for two simpler
objectives.
We generalize this approach from MDPs to full stochastic games.

Stochastic mean-payoff parity games were studied in \cite{CDGQ:2014}, where it was shown that
they can be solved in $\NP \cap \coNP$. However, this does not imply a solution for
stochastic energy-parity games, since, unlike in the non-stochastic
case \cite{CD2010},
there is no known reduction from energy-parity to mean-payoff parity
in stochastic games. (The reduction in \cite{CD2010} relies on the fact that
maximizer has a winning finite-memory strategy for energy-parity, which does
not generally hold for stochastic games or MDPs; see above.)

A related model are the 1-counter MDPs (and stochastic games) studied in
\cite{BBEKW2010,BBE2010,BBEK:IC2013}, since the value of the counter can be interpreted as
the stored energy. These papers consider the objective of reaching counter value
zero (which is dual to the energy objective of staying above zero), thus the roles
of minimizer and maximizer are swapped. However, unlike in this paper,
these works do not combine termination objectives with 
extra parity conditions.

\mypar{Structure of the paper.}
The rest of the paper is organized as follows. We start by introducing the notation and formal definitions of games and objectives in the next section. In \cref{sec:correctness} we show how checking almost-sure energy-parity objectives can be characterized in terms of two newly defined auxiliary objectives: Gain and Bailout.
In Sections \ref{sec:bailout} and \ref{sec:gain}, we show that almost-sure Bailout and Gain objectives, respectively, can be checked in \NP{} and \coNP.
\Cref{sec:main} contains our main result: \NP{} and \coNP{} algorithms for checking almost-sure energy-parity games with a known and unknown initial energy, as well as checking if a given initial energy is the minimal one. We conclude and point out some open problems in \Cref{sec:conclusion}. 
Due to page restrictions, most proofs in the main body of the paper were replaced by sketches.
The detailed proofs can be found in the appendix.

%% file: ex.two-ways.tex
\begin{figure}[b]
    \centering
    \includegraphics[width=0.5\linewidth]{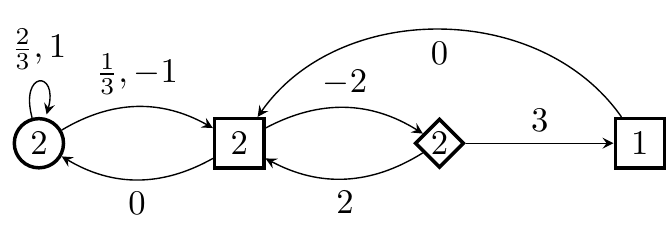}
    \caption{
        A SSG with two maximizer states ($\Max$), one minimizer state ($\Min$) and one probabilistic state ($\Ran$). Each state is annotated with its priority and each edge with a reward by which it increases the energy level (respectively, decreases if the reward is negative).
        The maximizer wins if the lowest priority visited infinitely often 
        is even and the energy level never drops below $0$.
    \label{fig:two-ways}}
\end{figure}

\begin{example}
  \label{ex:two-ways}
\Cref{fig:two-ways} shows an energy-parity game that the maximizer can win almost surely
when starting with an energy level of $\ge 2$ from the middle left node.
Whenever the game is at that node with an energy level $\ge 3$, 
then the maximizer can turn left and has at least $\frac{1}{2}$ chance that 
the energy level will never drop to $2$ while wining the game with
priority $2$.
This is because we can view this process as a random walk on a half line. If $x_n$ is the probability of reaching energy level $2$ when starting at $n$ then these probabilities are the least point-wise positive solution of the following system of linear equations: $x_2 = 1$, $x_n = \frac{2}{3} x_{n+1} + \frac{1}{3} x_{n-1}$ for all $n\geq 3$. We then get that $x_n = \frac{1}{2^{n-2}}$ so the probability of not reaching energy $2$ is $\geq \frac{1}{2}$ for all $n\geq 3$.
Always turning left guarantees that, almost surely, the parity condition holds and 
the limes inferior of the energy level is not $-\infty$.
We call this condition \emph{Gain}.
Strategies for \emph{Gain} can be used when the energy level is sufficiently
high (at least $3$ in our example) to win with a positive probability.

However, if maximizer plays for Gain and always moves left, then for every initial energy level
the chance of eventually dropping the energy down to level $2$
is positive, due to the negative cycle.
When that
happens, the only other option for the maximizer is to move right.
There minimizer can `choose how to lose', via a disjunction of two conditions
that we later formalize as \emph{Bailout}.
Either minimizer goes back to the start state without changing the
energy level (thus maximizer wins as the energy stays at level $2$
and only the good priority 2 is seen), or minimizer turns right.
In the latter case, the play visits a dominating odd priority
(which is bad for maximizer) but also increases the energy by $1$, which allows maximizer to switch back
to playing left for the \emph{Gain} condition until energy level $2$ is
reached again.

Our maximizer strategies are a complex interplay between \emph{Bailout} and \emph{Gain}.
In the example, it is easy to see that the probability of seeing priority $1$
infinitely often is zero if maximizer follows the just described strategy
(the probability of requiring to go right more than $n$ times is
at most $(\frac{1}{2})^n$), so maximizer wins this energy-parity game almost surely.
Note that maximizer does not win almost surely when the initial energy level is $0$ or $1$.
\end{example}

%% file: sec.prelim.tex
A probability distribution over a set $X$ is a function $f:X\to[0,1]$
such that $\sum_{x\in X} f(x) = 1$. We write $\Dist{X}$ for the set of
distributions over $X$.

\mypar{Games, Strategies, Measures.}
A \emph{Simple Stochastic Game (SSG)} is a 
directed graph
$\sys{G}\eqdef(V, E, \prob)$, where all states have an outgoing edge and 
the set of states is partitioned 
into
states owned by \emph{maximizer} ($\VMax$), \emph{minimizer} ($\VMin$) and probabilistic states ($\VRan$).
The set of \emph{edges} is $E\subseteq V\x V$ and $\prob:\VRan\to\Dist{E}$
assigns each probabilistic state a probability distribution over its outgoing edges. W.l.o.g., we assume that each probabilistic state has at most two successors, because one can introduce a new probabilistic state for each excess successor. We let $\prob(ws) \eqdef \prob(s)$ for all $ws \in (VE)^*\VRan$.

A \emph{path} is a finite or infinite sequence $\rho\eqdef s_0e_0s_1e_1\ldots$ such that $e_i = (s_i,s_{i+1})\in E$ holds for all indices $i$. A \emph{run} is an infinite path and we write $\Runs[]{}\eqdef (VE)^\omega$ for the set of all runs.

A \emph{strategy} for maximizer is a function $\sigma:(VE)^*\VMax\to\Dist{E}$ that assigns to each path $ws\in (VE)^*\VMax$ a probability distribution over the outgoing edges of its target node $s$. That is, $\sigma(ws)(e)>0$ implies $e=(s,t)\in E$ for some $t\in V$.
A strategy is called \emph{memoryless} if $\sigma(xs)=\sigma(ys)$ for all $x,y\in (VE)^*$ and $s\in \VMax$,
\emph{deterministic} if $\sigma(w)$ is Dirac for all $w\in (VE)^*\VMax$,
and \emph{finite-state} if there exists an equivalence relation $\sim$ on $(VE)^*\VMax$ with a finite index, such 
that $\sigma(\rho_1) = \sigma(\rho_2)$ if $\rho_1 \sim \rho_2$.
Of particular interest to us will be the class of \emph{memoryless deterministic strategies} (\emph{MD}) and
the class of \emph{finite-memory deterministic strategies} (\emph{FD}).
Strategies for minimizer are defined analogously and will usually be denoted by $\tau:(VE)^*\VMin\to\Dist{E}$.

A maximizing (minimizing) \emph{Markov Decision Process (MDP)} is 
a game in which minimizer (maximizer) has no choices, i.e., all her states have exactly one successor. We will write $\sys{G}[\tau]$ for the MDP resulting from fixing the strategy $\tau$.
A \emph{Markov chain} is a game where neither player has a choice.
In particular, $\sys{G}[\sigma,\tau]$ is a Markov chain obtained by 
setting, in the game $\sys{G}$, the strategies for maximizer and minimizer to $\sigma$ and $\tau$, respectively.

Given an initial state $s\in V$ and strategies $\sigma$ and $\tau$ for maximizer and minimizer, respectively, the set of runs starting in $s$ naturally extends to a probability space as follows.
We write $\Runs[\sys{G}]{w}$ for the \emph{$w$-cylinder}, i.e., the set of all runs with prefix $w\in (VE)^*V$. We let $\?F^\sys{G}$ be the $\sigma$-algebra generated by all these cylinders.
We inductively define a probability function $\Prob[G,\sigma,\tau]{s}$
on all cylinders,
which then uniquely extends 
to $\?F^\sys{G}$ by  Carath\'eodory's extension theorem~\cite{B1995},
by setting $\Prob[\sys{G},\sigma,\tau]{s}(\Runs[\sys{G}]{s}) \eqdef 1$
and
$\Prob[\sys{G},\sigma,\tau]{s}(\Runs[\sys{G}]{w}) \eqdef
\prod_{i=0}^{n-1} \mathit{dist}_i(s_0e_0s_1e_1\ldots s_{i})(e_i)$
for $w=s_0e_0s_1e_1\ldots e_{n-1}s_n$,
where $s_0 = s$, $e_i = (s_i, s_{i+1})$ and 
$\mathit{dist}_i$ is $\sigma(\cdot)$, $\tau(\cdot)$ or $\prob(\cdot)$,
for $s_i\in \VMax$,$\VMin$ or $\VRan$, respectively.
\mypar{Objective Functions.}
A (Borel) \emph{objective} is a set $\O\in \?F^\sys{G}$ of runs.
We write $\overline{\O}\eqdef\Runs{}\setminus\O$ for its complement.
Borel objectives $\O$ are weakly determined
\cite{M1998,Maitra-Sudderth:2003}, which means that
$$\sup_\sigma\inf_\tau\Prob[\sigma,\tau]{s}(\O)=
\inf_\tau\sup_\sigma\Prob[\sigma,\tau]{s}(\O).$$
This quantity is called the \emph{value}
of $\O$ in state $s$, and written as $\Val[\sys{G}]{s}{\O}$.
We say that $\O$ holds
\emph{almost-surely} (abbreviated as \emph{a.s.}) at state $s$ iff %
there exists $\sigma$ such that $\forall \tau, \Prob[\sys{G},\sigma,\tau]{s}(\O)=1$.
Let $\AS[\sys{G}]{\O}$ denote the set of states at which $\O$ holds almost surely.
We will drop the superscript $\sys{G}$ and simply write
$\Runs{}$, $\Prob[\sigma,\tau]{s}$ and $\AS{\O}$,
if the game is clear from the context.

We use the syntax and semantics of operators $\eventually$ (eventually)
and $\always$ (always) from the temporal logic LTL \cite{CGP:book} to specify
some conditions on runs.

\smallskip
A \emph{reachability condition} is defined by a set of target states $T \subseteq V$. 
A run $\rho=s_0e_0s_1\ldots$ satisfies the reachability condition iff
there exists an $i \in \N$ s.t.\ $s_i \in T$.
We write $\eventually T \subseteq \Runs{}$ for the set of runs that satisfy this
reachability condition.
Given a set of states $W \subseteq V$, we lift this to a safety condition on
runs and write $\always W \subseteq \Runs{}$ for the set of runs $\rho=s_0e_0s_1\ldots$ where $\forall
i.\, s_i \in W$.

\smallskip
A \emph{parity condition} is given by a bounded function $\parity:V\to\N$
that assigns a priority (a non-negative integer) to each state.
A run $\rho \in \Runs{}$ satisfies the parity condition iff
the minimal priority that appears infinitely often
on the run is even.
The \emph{parity objective} is the subset $\Parity[]{} \subseteq
\Runs{}$ of runs that satisfy the parity condition.

\smallskip
\emph{Energy conditions} are given by a function $\cost{}:E\to\Z$,
that assigns a \emph{reward} value to each edge.
For a given initial energy value $k\in\N$, a run $s_0e_0s_1e_1\ldots$ satisfies the $k$-energy condition
if, for every finite prefix of length $n$, the \emph{energy level} $k+\sum_{i=0}^n\cost(e_i)$ is greater or equal to $0$.
Let $\EN{k} \subseteq \Runs{}$ denote the $k$-energy objective,
consisting of those runs that satisfy the $k$-energy condition.

The \emph{$l$-storage condition} holds for a run $s_0e_0s_1e_1\ldots$
if
$l+\sum_{i=m}^{n-1}\cost(s_i,s_{i+1})\ge 0$
holds for every infix $s_m e_m s_{m+1}\ldots s_n$.
Let $\ES{k,l} \subseteq \Runs{}$ denote the $k$-energy $l$-storage objective,
consisting of those runs that satisfy both the $k$-energy
and the $l$-storage condition.
We write $\ES{k}$ for $\bigcup_l \ES{k,l}$. Clearly, $\ES{k} \subseteq \EN{k}$.

\smallskip
\emph{Mean-payoff} and \emph{limit-payoff conditions} are defined w.r.t.~the same reward function
as the energy conditions. 
The \emph{mean-payoff} value
of a run $\rho=s_0e_0s_1e_1\ldots$ is
$\mathit{MP}(\rho)\eqdef \liminf_{n\rightarrow\infty}\frac{1}{n}\sum_{i=0}^{n-1}\cost(e_i)$.
For $\constraint\in\{>,\ge,=,\le,<\}$ and $c\in\R\cup\{-\infty,\infty\}$,
the set $\MP{\constraint c}\subseteq \Runs{}$ consists of all runs $\rho$ with $\mathit{MP}(\rho)\constraint c$.
Let $\LimInf{\constraint c}\subseteq \Runs{}$ contain all runs $\rho$ with
$(\lim\inf_{n\to\infty} \sum_{i=0}^{n}\cost(e_i))\constraint c$,
and likewise for $\LimSup{\constraint c}$.

The combined energy-parity objective $\EN{k}\cap\Parity{}$ is Borel and therefore weakly determined,
meaning that it has a well-defined ($\inf\sup=\sup\inf$) value for every game \cite{M1998,Maitra-Sudderth:2003}.
Moreover, the almost-sure energy-parity
objective (asking to win with probability $1$)
is even strongly determined \cite{KMSW2017}: either maximizer has a
strategy to enforce the condition with probability $1$ or minimizer has a
strategy to prevent this.

%% file: sec.characterization.tex
The main theorem of this section (\cref{thm:correctness})
characterizes almost sure energy-parity objectives in terms of two intermediate
objectives called $\Gain$ and $\kBailout$ for parameters $k \geq 0$.
This will form the basis of all computability results: we will show
(as \cref{thm:bailout-comp,thm:coNP-gain,thm:as-gain-np})
how
to compute almost-sure sets for these intermediate objectives.

\begin{definition}\label{def:gain-bailout}
Consider a finite SSG $\?G=(V,E,\prob)$, as well as reward and parity
functions defining the objectives $\Parity,\LimInf{>\infty},\LimSup{=\infty}$
as well as $\ES{k,l}$ and $\EN{k}$ for every $k,l\in\N$.
We define combined objectives $\Gain$ and $\kBailout\eqdef\cup_{l} \Bailout(k,l)$ where %
\begin{align*}
    \Gain         &\quad\eqdef\quad \LimInf{>-\infty}\cap\Parity\\
    \Bailout(k,l) &\quad\eqdef\quad (\ES{k,l}\cap\Parity)\cup(\EN{k}\cap\LimSup{=\infty}).
\end{align*}
\end{definition}

The main idea behind these two objectives is a special witness property
for energy-parity.
We argue that, if maximizer has an almost-sure winning strategy for energy-parity then he also has one that combines two almost-sure winning strategies, one for $\Gain$ and one for $\kBailout$.

Notice that playing an almost-sure winning strategy for $\Gain$
implies a uniformly lower-bounded strictly positive
chance that the energy level never drops below zero (assuming it is sufficiently high to begin with).
This fact uses the finiteness of the set of
control-states and does not hold for infinite-state MDPs.
In the unlikely event that the energy level does get close to zero, maximizer switches to playing an almost sure winning strategy for $\kBailout$.
This is a disjunction of two scenarios, and the balance might be influenced
by minimizer's choices. In the first scenario $(\ES{k,l}\cap\Parity)$
the energy never drops much and stays above zero (thus satisfying
energy-parity).
In the second scenario, $(\EN{k}\cap\LimSup{=\infty})$, the parity objective is
temporarily suspended in favor of boosting 
(while always staying above zero)
the energy to a sufficiently high
level to switch back to the strategy for $\Gain$ and thus try again from the beginning.
The probability of infinitely often switching between these modes is zero
due to the lower-bounded chance of success in the $\Gain$ phase.
Therefore, maximizer eventually wins by playing for $\Gain$.
Note that maximizer needs to remember the current energy level
in order to know when to switch and consequently, this strategy uses infinite memory. 

\begin{example}
    \label{ex:two-ways-cont}
Consider again the game in \cref{fig:two-ways}.
The middle left state satisfies both $\Gain$ and $\kBailout$
objectives for all $k\geq 2$ almost-surely. The respective winning strategies are
to always go left for $\Gain$ or always go right for $\kBailout$ when at that state.
Note that it neither satisfies $0$-$\Bailout$ nor $1$-$\Bailout$ objectives.
\end{example}

We define the subset $W \subseteq V$ of states from which maximizer can almost
surely win both $\Gain$ and $\kBailout$ (assuming sufficiently high initial energy),
while at the same time ensuring that the play remains within this set of states.
These are the states from which maximizer can win
by freely combining individual strategies for the $\Gain$\ and $\Bailout$\ objectives.

\begin{definition}\label{def:W}
Given a finite SSG $\?G=(V,E,\prob)$,
let $W \subseteq V$ be the largest subset of states satisfying the following condition
\[
W \subseteq \AS{\Gain \cap \always W}\, \cap\,
\bigcup_{k}\AS{\kBailout \cap \always W}
\]
\end{definition}

This condition describes a fixed-point, and as it is easy to see that 
if two sets $W_1$ and $W_2$ are such fixed-points, then so is $W_1 \cup W_2$.
Thus, the maximal fixed-point $W$ is well-defined.

\medskip
Our main characterization
of almost-sure energy-parity objectives
is the following \cref{thm:correctness}.
It states that maximizer can almost surely win 
an $\EN{k} \cap \Parity$ objective if, and only if,
he can win the easier $k$-Bailout objective
while always staying in the safe set $W$.

\begin{restatable}{theorem}{thmCorrectness}
    \label{thm:correctness}
For every $k\in\N$,
$
\AS{\EN{k} \cap \Parity} = \AS{\kBailout \cap \always W}
$.
\end{restatable}

Our proof of this characterization theorem relies
on the following 
claim, which allows to lift the existence
of finite-memory deterministic optimal strategies 
from MDPs to SSGs.
It applies to a fairly general class of objectives and, we believe, is of independent interest.

\medskip
Recall that $\overline{\O}\eqdef\Runs{}\setminus\O$
denotes the complement of objective $\O$.
For runs $a,b,c\in\Runs{}$ we say that $a$ is a \emph{shuffle} of
$b$ and $c$ if there exist factorizations
$b=b_0b_1\dots$
and
$c=c_0c_1\dots$
such that 
$a=b_0c_0b_1c_1\dots$.
An objective $\O$ is called \emph{submixing} if, 
for every run $a\in \O$ that is a shuffle of runs $b$ and $c$,
either $b\in\O$ or $c\in\O$.
$\O$ is \emph{shift-invariant} if,
for every run $s_1e_1s_2e_2\ldots$, it holds that $s_1e_1s_2e_2\ldots \in \O \iff s_2e_2\ldots \in \O$.
Shift-invariance slightly generalizes the better-known \emph{tail} condition (see \cite{GK2014} for a discussion).

\begin{restatable}{theorem}{lemFDstrategy}
\label{lem:FD-strategy}
    Let 
    $\?O$ be an objective
    such that $\overline{\?O}$ is both shift-invariant and submixing.
    If maximizer has optimal FD strategies (from any state $s$) for $\?O$ 
    for every finite MDP %
    then 
    maximizer has optimal FD strategies (from any state $s$) for $\?O$ 
    for every finite SSG. %
\end{restatable}

This applies in particular to the $\Gain$ objective, but not to $\kBailout$ objectives, as these are not shift-invariant.
A proof of \cref{lem:FD-strategy} can be found in \cref{app:lifting}.
It uses a recursive argument based on the notion of \emph{reset strategies}
from \cite{GK2014}.

%% file: proof.correctness.tex
The remainder of this section is dedicated to proving \cref{thm:correctness}.
We will first collect the remaining technical claims
about $\Gain$, $\Bailout$, and reachability objectives.
Most notably, as \cref{lem:gain-value},
we show that if maximizer can almost surely win $\Gain$ in a SSG,
then he can do so using a FD strategy which moreover satisfies an energy-parity objective
with strictly positive (and lower-bounded) probability.
This is shown in part based on \cref{lem:FD-strategy} applied to the $\Gain$ objective.
We will also need the following fact about reachability objectives in finite MDPs.

\begin{lemma}[{\cite[Lemma~3.9]{BBEK:IC2013}}]\label{BBEK:IC2013:Lemma3.9}
Let $\?M$ be a finite MDP and 
${\it Reach}_T$ be the reachability objective %
with target $T \eqdef \{s'\ |\ \Val{s'}{\LimInf{=-\infty}}=1\}$.
One can compute a rational constant $c< 1$
and an integer $h \ge 0$ such that for all states $s$ and 
$i \ge h$ we have
$
\forall \tau.\, \Prob[\tau]{s}(\overline{\EN{i}} \cap \overline{{\it Reach}_T}) \le \frac{c^i}{1-c}
$.
\end{lemma}

\begin{lemma}\label{lem:gain-value}
Consider a finite SSG $\?G=(V,E,\prob)$ 
where $\Gain$ holds a.s.\ for every state $s\in V$.
Then, for every $\delta \in [0,1)$ and $s \in V$, there exists a
$\hat{k} \in \N$
and an FD strategy $\hat{\sigma}$
s.t. 
\begin{enumerate}
\item
$\forall \tau.\, \Prob[\hat{\sigma},\tau]{s}(\Gain)=1$, and \label{lem:gain-value-ad1}
\item
$\forall \tau.\, \Prob[\hat{\sigma},\tau]{s}(\EN{\hat{k}} \cap \Parity) \ge \delta$.
\label{lem:gain-value-ad2}
\end{enumerate}
\end{lemma}
\begin{proof}
Fix a $\delta \in [0,1)$ and a state $s \in V$.
Both $\LimInf{=-\infty}$, as well as $\Parity$ objectives are 
\emph{shift-invariant} and \emph{submixing},
and therefore also the union has both these
properties.
It follows
that $\overline{\Gain} = \overline{\LimInf{>-\infty}\cap\Parity} = \LimInf{=-\infty}\cup\overline{\Parity}$
is both shift-invariant and submixing, since the complement of a parity
objective is also a parity objective.
By \cref{lem:MDP-FD-strategy} and \cref{lem:FD-strategy}, there exists an almost-sure winning FD
strategy $\hat{\sigma}$ for maximizer for the objective $\Gain$ %
from $s$, 
i.e., 
$\forall \tau.\, \Prob[\hat{\sigma},\tau]{s}(\Gain)=1$,
thus yielding \cref{lem:gain-value-ad1}.

Let $\?M$ be the MDP obtained from $\?G$ by fixing the strategy
$\hat{\sigma}$ for maximizer from $s$.
Since $\?G$ is finite and $\hat{\sigma}$ is FD, also $\?M$ is finite.  
In $\?M$ we have
$\forall \tau.\, \Prob[\tau]{s}(\Gain)=1$.
In particular, in $\?M$, the set $T \eqdef \{s'\ |\ \Val{s'}{\LimInf{=-\infty}}=1\}$
is not reachable, i.e., 
$\forall \tau.\, \Prob[\tau]{s}({\it Reach}_T)=0$.

By \cref{BBEK:IC2013:Lemma3.9}, in $\?M$
there exists a horizon $h \in\N$ and a constant $c <1$ such that for all $i \ge h$ we have
$\forall \tau.\, \Prob[\tau]{s}(\overline{\EN{i}} \cap \overline{{\it Reach}_T}) \le \frac{c^i}{1-c}$.
Since $T$ cannot be reached in $\?M$, the condition $\overline{{\it Reach}_T}$ evaluates 
to ${\it true}$ and we have
$\forall \tau.\, \Prob[\tau]{s}(\EN{i}) \ge 1-\frac{c^i}{1-c}$.
Since $c<1$ and $\delta<1$, we can pick a sufficiently large $\hat{k} \ge h$ such that 
$1-\frac{c^{\hat{k}}}{1-c} \ge \delta$ and obtain
$\forall \tau.\, \Prob[\tau]{s}(\EN{\hat{k}}) \ge \delta$ in $\?M$.
Moreover, the above property 
$\forall \tau.\, \Prob[\tau]{s}(\Gain)=1$
in particular implies 
$\forall \tau.\, \Prob[\tau]{s}(\Parity)=1$.
Thus we obtain
$\forall \tau.\, \Prob[\tau]{s}(\EN{\hat{k}} \cap \Parity) \ge \delta$ in $\?M$.

Back in the SSG $\?G$, we have
$\forall \tau.\, \Prob[\hat{\sigma},\tau]{s}(\EN{\hat{k}} \cap \Parity) \ge \delta$
as required for \cref{lem:gain-value-ad2}.
\end{proof}

\begin{lemma}\label{lem:enpar-bailout}
  $\EN{k} \cap \Parity \subseteq \kBailout$.
\end{lemma}
\begin{proof}
Let $\rho$ be a run in $\EN{k} \cap \Parity$.
There are two cases. In the first case we have $\rho \in \cup_{l} \ES{k,l}\cap\Parity$ and thus
directly $\rho \in \kBailout$.
Otherwise, $\rho \notin \cup_{l} \ES{k,l}\cap\Parity$.
Since $\rho \in \Parity$, we must have $\rho \notin \cup_{l} \ES{k,l}$.
Since $\rho \in \EN{k}$, it follows that $\rho$ does not
satisfy the $l$-storage condition for any $l \in \N$.
So, for every $l \in \N$, there exists an infix $\rho'$ of $\rho$ 
s.t.\ $l + \cost(\rho') < 0$.
Let $\rho''$ be the prefix of $\rho$ before $\rho'$.
Since $\rho \in \EN{k}$ we have $k + \cost(\rho''\rho') \ge 0$
and thus $\cost(\rho'') \ge -k - \cost(\rho') > -k + l$.
To summarize, if $\rho \notin \cup_{l} \ES{k,l}\cap\Parity$
then, for every $l$, it has a prefix $\rho''$ with
$\cost(\rho'') > -k + l$. Thus
$\rho \in \LimSup{=\infty}$.
Thus $\rho \in \kBailout$.
\end{proof}

We now define $W'$ as the set of states that are almost-sure winning for
energy-parity with some sufficiently high initial energy level.
($W'$ is also called the winning set for the unknown initial credit problem.) 

\begin{definition}\label{def:Wprime}
$W' \eqdef \bigcup_{k} \AS{\EN{k} \cap \Parity}$.
\end{definition}

\begin{lemma}\label{lem:connect-aspar-gain-bailout}
\
\begin{enumerate}
\item\label{lem:connect-aspar-gain-bailout-1}
$\AS{\EN{k} \cap \Parity} \subseteq \AS{\Gain \cap \always W'}$
\item\label{lem:connect-aspar-gain-bailout-2}
$\AS{\EN{k} \cap \Parity} \subseteq \AS{\kBailout \cap \always W'}$
\end{enumerate}
\end{lemma}
\begin{proof}
Let $s \in \AS{\EN{k} \cap \Parity}$ and
$\sigma$ a strategy that witnesses this property.
Except for a null-set, all runs
$\rho = se_0s_1e_1\ldots e_{n-1}s_n \dots$ from $s$ induced by $\sigma$ satisfy
$\EN{k} \cap \Parity$. 

Let $\rho' = se_0s_1e_1\ldots s_m$ be a finite prefix of $\rho$.
For every $n \ge 0$ we have
$k+\sum_{i=0}^{n-1} \cost(e_i) \ge 0$, since $\rho \in \EN{k}$.
In particular this holds for all $n \ge m$.
So, for every $n \ge m$, we have 
$k+\sum_{i=0}^{m-1} \cost(e_i) + \sum_{i=m}^{n-1} \cost(e_i) \ge 0$.
Therefore $s_m \in \AS{\EN{k'} \cap \Parity}$, where
$k' = k+\sum_{i=0}^{m-1} \cost(e_i)$, as witnessed by playing $\sigma$
with history $se_0s_1e_1\ldots s_m$ from $s_m$.
Thus $s_m \in \bigcup_{k} \AS{\EN{k} \cap \Parity} = W'$, i.e., 
almost all $\sigma$-induced runs $\rho$ satisfy $\always W'$.

Towards \cref{lem:connect-aspar-gain-bailout-1},
we have $\EN{k} \subseteq \LimInf{>-\infty}$ and
thus $\EN{k} \cap \Parity \subseteq \LimInf{>-\infty}\cap\Parity = \Gain$.
Therefore $\sigma$ witnesses
$s \in \AS{\Gain \cap \always W'}$.

Towards \cref{lem:connect-aspar-gain-bailout-2},
we have $\EN{k} \cap \Parity \subseteq \kBailout$
by \cref{lem:enpar-bailout}.
Thus $\sigma$ witnesses $s \in \AS{\kBailout \cap \always W'}$.
\end{proof}

\begin{lemma}\label{claim1}
$W' \subseteq W$.
\end{lemma}
\begin{proof}
It suffices to show that $W'$ satisfies the monotone condition imposed on
$W$ (cf. \cref{def:W}), since $W$ is defined as the largest set satisfying this condition.

Let $s \in W' = \bigcup_{k} \AS{\EN{k} \cap \Parity}$.
Then $s \in \AS{\EN{\hat{k}} \cap \Parity}$ for some fixed $\hat{k}$.
By \cref{lem:connect-aspar-gain-bailout}(1) we have
$s \in \AS{\Gain \cap \always W'}$.
By \cref{lem:connect-aspar-gain-bailout}(2) we have
$s \in \AS{\hat{k}\text{-}\Bailout \cap \always W'}
\subseteq
\bigcup_{k}\AS{\kBailout \cap \always W'}$.
\end{proof}

\begin{proof}[Proof of \cref{thm:correctness}]
Towards the $\subseteq$ inclusion,
we have
\[
\AS{\EN{k} \cap \Parity}
\subseteq
\AS{\kBailout \cap \always W'}
\subseteq
\AS{\kBailout \cap \always W}
\]
by \cref{lem:connect-aspar-gain-bailout}(2)
and
\cref{claim1}.

Towards the $\supseteq$ inclusion, let 
$s \in \AS{\kBailout \cap \always W}$ and $\sigma_1$ be a strategy that witnesses this.
We show that 
$s \in \AS{\EN{k} \cap \Parity}$.
We now consider the modified SSG $\?G' = (W,E,\prob)$ with the state set
restricted to $W$. In particular, $s \in W$ and 
$\sigma_1$ witnesses $s \in \AS{\kBailout}$ in $\?G'$.
We now construct a strategy $\sigma$ that witnesses 
$s \in \AS{\EN{k} \cap \Parity}$ in $\?G'$, and thus also in $\?G$.
The strategy $\sigma$ will use infinite memory to keep track of the
current energy level of the run.

Apart from $\sigma_1$, we require several more strategies as building blocks
for the construction of $\sigma$.

First, in $\?G$ we had $\forall s' \in W.\, s' \in \AS{\Gain \cap \always W}$, 
and thus in $\?G'$ we have $\forall s' \in W.\, s' \in \AS{\Gain}$.
For every $s' \in W$ we instantiate \cref{lem:gain-value} for $\?G'$ with
$\delta = 1/2$ and obtain a number $\hat{k}_{s'}$ and a strategy
$\hat{\sigma}_{s'}$
with
\begin{enumerate}
\item
$\forall \tau.\, \Prob[\hat{\sigma}_{s'},\tau]{s'}(\Gain)=1$, and \label{thm:correctness-ad1}
\item
$\forall \tau.\, \Prob[\hat{\sigma}_{s'},\tau]{s'}(\EN{\hat{k}_{s'}} \cap \Parity) \ge 1/2$.
\label{thm:correctness-ad2}
\end{enumerate}
Let $k_1 \eqdef \max\{\hat{k}_{s'}\ |\ s' \in W\}$.
The strategies $\hat{\sigma}_{s'}$ are called \emph{gain strategies}.

Second, by the finiteness of $V$, there is a minimal number
$k_2$ such that $\bigcup_k \AS{\kBailout \cap \always W} = \bigcup_{k \le k_2}
\AS{\kBailout \cap \always W}$ in $\?G$.
Therefore, in $\?G'$ we have that
$$W \subseteq \bigcup_k \AS{\kBailout} = \bigcup_{k \le k_2}
\AS{\kBailout} = \AS{k_2\text{-}\Bailout}.$$
Thus in $\?G'$ for every $s' \in W$ there exists a strategy
$\tilde{\sigma}_{s'}$ with
$\forall \tau.\, \Prob[\tilde{\sigma}_{s'},\tau]{s'}(k_2\text{-}\Bailout)
= 1$.
The strategies $\tilde{\sigma}_{s'}$ are called \emph{bailout strategies}.
Let $k' \eqdef k_1 + k_2 - k +1$. We now define the strategy $\sigma$.
\begin{description}
\item[Start:] 
First $\sigma$ plays like $\sigma_1$ from $s$. 
Since $\sigma_1$ witnesses $s \in \AS{\kBailout}$ 
against every minimizer strategy $\tau$, almost all induced runs
$\rho = s e_0s_1e_1\ldots$ satisfy either
\begin{description}
\item[(A)] $(\cup_{l}\ES{k,l} \cap \Parity)$, or
\item[(B)] $(\EN{k} \cap \LimSup{=\infty})$.
\end{description}
Almost all runs $\rho$ of the latter type (B) (and potentially also some runs of type
(A)) satisfy $\EN{k}$ and $\sum_{i=0}^l\cost(e_i) \ge k'$ eventually for some
$l$. 
If we observe $\sum_{i=0}^l\cost(e_i) \ge k'$ for some prefix 
$s e_0s_1e_1\ldots e_l s'$ of the run $\rho$ 
then our strategy $\sigma$ plays from $s'$ as
described in the {\bf Gain} part below.
Otherwise, if we never observe this condition, then our run $\rho$
is of type (A) and $\sigma$ continues playing like $\sigma_1$.
Since property (A) implies $(\EN{k} \cap \Parity)$, this is sufficient.
\item[Gain:]
In this case we are in the situation where we have reached 
some state $s'$ after some finite prefix $\rho'$ of the run, where
$\cost(\rho') \ge k'$. Our strategy $\sigma$ now 
plays like the gain strategy $\hat{\sigma}_{s'}$, as long as
$\cost(\rho') \ge k'-k_1$ holds for the current prefix $\rho'$ of the run.
By \cref{thm:correctness-ad2}, this will satisfy 
$\forall \tau.\, \Prob[\hat{\sigma}_{s'},\tau]{s'}(\EN{\hat{k}_{s'}} \cap
\Parity) \ge 1/2$ and thus 
$\forall \tau.\, \Prob[\hat{\sigma}_{s'},\tau]{s'}(\EN{k_1} \cap \Parity) \ge 1/2$.
It follows that with probability $\ge 1/2$ we will keep playing
$\hat{\sigma}_{s'}$ forever and satisfy $\Parity$ and always $\cost(\rho') \ge k'-k_1$
and thus $\EN{k}$, since $k + \cost(\rho') \ge k+k'-k_1 = k_2 + 1\ge 0$.

Otherwise, if eventually $\cost(\rho') = k'-k_1 -1$ then we have 
$k + \cost(\rho') = k_2$. In this case (which happens with probability $<1/2$)
we continue playing as described in the {\bf Bailout} part below.
\item[Bailout:]\label{item:sigma:bailout}
In this case we are in the situation where we have reached 
some state $s'' \in W$ after some finite prefix $\rho'$ of the run, where
$k+\cost(\rho') = k_2$. 
Since $s'' \in W$, we can now let our strategy $\sigma$ 
play like the bailout strategy $\tilde{\sigma}_{s''}$ and obtain
$\forall \tau.\,
\Prob[\tilde{\sigma}_{s''},\tau]{s''}(k_2\text{-}\Bailout) = 1$.
Thus almost all induced runs $\rho'' = s'' e_0s_1e_1\ldots$ from $s''$ 
satisfy either
\begin{description}
\item[(A)] $(\cup_{l}\ES{k_2,l} \cap \Parity)$, or
\item[(B)] $(\EN{k_2} \cap \LimSup{=\infty})$.
\end{description}
As long as $\cost(\rho') < k'$ holds for the current prefix $\rho'$ of the run, we
keep playing $\tilde{\sigma}_{s''}$. 
Otherwise, if eventually $\cost(\rho') \geq k'$ holds, then we switch back to playing
the {\bf Gain} strategy above.
All the runs that never switch back to playing
the {\bf Gain} strategy must be of type (A) and thus satisfy $\Parity$.
Since we have 
$k_2\text{-}\Bailout \subseteq \EN{k_2}$, it follows that,
for every prefix $\rho''$ of the run from $s''$, according to $\tilde{\sigma}_{s''}$
we have $k_2 + \cost(\rho'') \ge 0$.
Thus, for every prefix $\rho'''$ of $\rho$, we have
$k + \cost(\rho''') = k + \cost(\rho') + \cost(\rho'') = k_2 + \cost(\rho'') \ge 0$.
Therefore, the $\EN{k}$ objective is satisfied by all runs.
\end{description}
As shown above, almost all runs induced by $\sigma$ that eventually stop switching between the three modes
satisfy $\EN{k} \cap \Parity$.
Switching from Gain/Bailout to Start is impossible, but switching from Gain to
Bailout and back is possible. However, the set of runs that infinitely often
switch between Gain and Bailout is a null-set, because the probability of
switching from Gain to Bailout is $\le 1/2$.
Thus, $\sigma$ witnesses $s \in \AS{\EN{k} \cap \Parity}$. 
\end{proof}

\begin{remark}\label{rem:WvsWprime}
It follows from the results above that $W' = W$.
The $\subseteq$ inclusion holds by \cref{claim1}.
For the reverse inclusion we have
\begin{align*}
W &\subseteq  \bigcup_{k}\AS{\kBailout \cap \always W} & \text{by \cref{def:W}}\\
  &=          \bigcup_{k}\AS{\EN{k} \cap \Parity}      & \text{by \cref{thm:correctness}}\\
  &=          W'                                      & \text{by \cref{def:Wprime}.}
\end{align*}
\end{remark}

%% file: sec.bailout.tex
In this section we will argue that it is possible 
decide, in \NP\ and \coNP, whether the bailout objective can be satisfied almost surely. %
More precisely, we show the existence of procedures to decide if, for a given $k\in\N$ and state $s$, there exists an $l\in\N$ such that
$s$ almost-surely satisfies the $\Bailout(k,l)$ objective

$$ \Bailout(k,l)\quad \eqdef \quad(\ES{k,l}\cap\Parity)\cup(\EN{k}\cap\LimSup{=\infty}). $$

Recall that the idea behind the Bailout objective is
that, during a game for energy-parity, maximizer is temporarily abandoning the parity (but not the energy) condition in order to increase the energy to a sufficient level (which will then allow him to try an a.s.~strategy for $\Gain$\ once more).
However, in a stochastic game -- as opposed to an MDP \cite{MSTW2017} --
an opponent could possibly prevent this increase in energy level at the expense of satisfying the
original energy-parity objective in the first place (cf. Example \ref{ex:two-ways}).
The Bailout objective is designed to capture the disjunction of both outcomes,
as both are favorable for the maximizer.
The parameter $k$ is the acceptable total energy drop (i.e., the initial value),
and the parameter $l$ is the acceptable energy drop on any infix of a play,
which translates to the upper bound on the energy level in the second outcome.

\medskip
The question can be phrased equivalently as 
membership of a control state $s$ in 
the almost-sure set for the $\kBailout$ objective %
for a given game $\sys{G}$ and energy level $k\in\N$.

\begin{restatable}{theorem}{thmBailoutComp}
    \label{thm:bailout-comp}
    One can check in $\NP, \coNP$ and pseudo-polynomial time
    if, for a given SSG $\sys{G}\eqdef(V, E, \prob)$,
    $k\in\N$ and control state $s\in V$,
    maximizer can almost-surely satisfy
    ${\kBailout}$
    from $s$.
    
    Moreover, there are 
    $K,L\in\N$, polynomial in $\card{V}$ and the largest absolute transition reward,
    so that
    $\bigcup_{k\ge 0}\AS[\sys{G}]{\kBailout} = \AS[\sys{G}]{\Bailout(K,L)}$.
    And so, checking whether state $s$ belongs to $\bigcup_{k\ge 0}\AS[\sys{G}]{\kBailout}$ is in $\NP$ and $\coNP$.
\end{restatable}

\begin{proof}[Proof (sketch)]
This is shown by a sequence of transformations of the game and ultimately reduced to a finding the winner of a non-stochastic game with an  energy-parity objective, which is known to be solvable in $\NP, \coNP$ and pseudo-polynomial time \cite{CD2012}.
One important observation is that it is possible to replace, without changing the outcome,
the energy $\EN{k}$ condition in the $\Bailout(k,l)$ objective
by the more restrictive energy-storage $\ES{k,l}$ condition.
See \cref{app:bailout} for further details.
\end{proof}

%% file: sec.gain.tex
In this section we will argue that it is possible to decide, in \NP\ and \coNP, whether the $\Gain$ objective (i.e., $\LimInf{>-\infty}\cap\Parity$) can be satisfied almost surely.

We start by investigating the strategy complexity
of winning strategies for the $\Gain$ objective.

\begin{restatable}{lemma}{lemgainminMD}
\label{lem:gain:min-MD}
In every finite SSG, minimizer has optimal MD strategies for objective $\Gain$.
\end{restatable}
\begin{proof}
    We show that maximizer has MD optimal strategies for $\LimInf{=-\infty}\cup\Parity$.
    This is equivalent to the claim of the lemma because
    $\overline{\LimInf{>-\infty}\cap\Parity}=\LimInf{=-\infty}\cup\overline{\Parity}$
    and the complement of a parity condition is itself a parity condition (with all priorities incremented by one).

    We note that both $\LimInf{=-\infty}$, as well as parity objectives $\Parity{}$ are 
    shift-invariant and {submixing}
    and therefore also that the union
    $\LimInf{=-\infty}\cup\Parity$ has both these properties.
    The claim now follows from the fact that
    SSGs with objectives that are both submixing and shift-invariant
    admit MD optimal strategies for maximizer \cite[Theorem 5.2]{GK2014}.
\end{proof}

\noindent
Based on the results in \cite{MSTW2017} %
one can show a similar claim for maximizer strategies in MDPs.

\begin{restatable}{lemma}{lemMDPFDstrategy}
\label{lem:MDP-FD-strategy}
For finite MDPs, almost-sure winning maximizer strategies for $\Gain$ can be chosen FD.
\end{restatable}

\noindent
\input{sec.gain.coNP}

\input{sec.gain.NP}

%% file: sec.gain.coNP.tex
Using the existence of MD optimal minimizer strategies (\cref{lem:gain:min-MD})
and a $\coNP$ upper bound for checking almost sure $\Gain$ in MDPs established in \cite{MSTW2017},
we can derive a \coNP\ procedure.
See Appendix \ref{app:gain-coNP} for full details.

\begin{restatable}{theorem}{coNPgain}
    \label{thm:coNP-gain}
    Checking whether a state $s\in V$ of a SSG satisfies $\Gain$ almost-surely is in $\coNP$.
\end{restatable}

%% file: sec.gain.NP.tex
The rest of this section will deal with the $\NP$ upper bound, which is the most challenging part of this paper. 
The crux of our proof is the observation that if maximizer has a strategy that wins almost surely against all MD minimizer strategies, then he wins almost surely.
This is because one of these MD strategies is optimal due to \cref{lem:gain:min-MD}.
We show that, in order to witness such an almost-sure winning strategy for maximizer in SSG $\sys{G}$, it suffices to
provide a polynomially larger SSG $\sys{G}_3$, together with an almost-sure winning strategy for the \emph{storage-parity} objective (see \Cref{thm:storage-parity-comp} in \Cref{sec:main}) in $\sys{G}_3$.
This will give us an $\NP$ algorithm, because $\sys{G}_3$, along with its winning strategy, can be guessed and verified in polynomial time.
Formally we claim that:

\begin{restatable}{theorem}{thmASGainNP}
\label{thm:as-gain-np}
    Checking whether a state $s\in V$ of $\?G$ satisfies $\Gain$ almost-surely is in $\NP$.
\end{restatable}

\begin{proof}
    (sketch)
For technical convenience, we will assume w.l.o.g.\ that every SSG henceforth is in a normal form,
where every random state has only one predecessor, which is owned by the maximizer.
To show the existence of $\sys{G}_3$, we are going to introduce two intermediate games:
$\sys{G}_1$ and $\sys{G}_2$. 
These games are never constructed by our $\NP$ algorithm, but are just defined to break down the complex construction of $\sys{G}_3$ into more manageable steps. 

Intuitively, $\sys{G}_1$ is just $\sys{G}$ where all rewards on edges are multiplied by a large enough factor, $f$, to turn strategies with a mean-payoff $>0$ into ones with mean-payoff $>2$. $\sys{G}_2$ is an extension of $\sys{G}_1$ where the maximizer is given a choice before every visit to a probabilistic node. He can either let the game proceed as before,
or sacrifice part of his one-step reward in exchange for a more evenly balanced reward outcome,
so the energy can no longer drop arbitrarily low when a probabilistic cycle is reached. 
As a result, in $\sys{G}_2$ it suffices to consider a storage-parity objective (see \Cref{thm:storage-parity-comp} in \Cref{sec:main}) instead of $\Gain$.
The number of choices maximizer is given is the number of MD minimizer strategies, which clearly can be exponential.
That would not suffice for an $\NP$ algorithm. Therefore, we show that most of these choices are redundant and can be removed without impairing the almost sure wining region. As the result of that pruning, we obtain $\sys{G}_3$ of polynomial size.
\end{proof}

For the the technical details of the
$\sys{G} \to \sys{G}_1 \to \sys{G}_2 \to \sys{G}_3$ constructions please see \cref{app:gain-NP}.
\Cref{fig:np} shows how these transformations may look like.

\begin{figure}[htb!]
    \begin{subfigure}[t]{0.45\textwidth}
        \centering
        \includegraphics[width=0.95\textwidth]{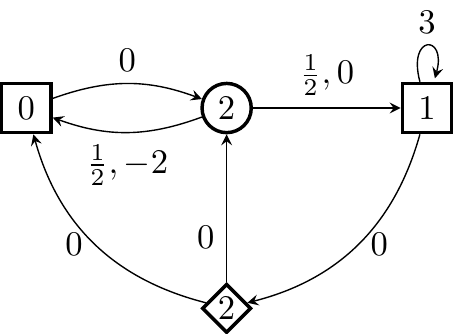}
        \caption{The original game $\sys{G}=\sys{G}_1$}
        \label{fig:np:a}
    \end{subfigure}\hfill
    \begin{subfigure}[t]{0.55\textwidth}
        \centering
        \includegraphics[width=0.95\textwidth]{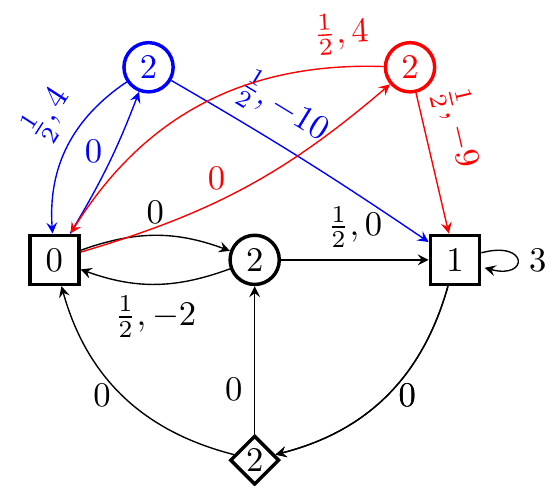}
        \caption{The game $\sys{G}_2$}
        \label{fig:np:b}
    \end{subfigure}\hfill
    \centering
       \begin{subfigure}[t]{0.45\textwidth}
        \includegraphics[width=\textwidth]{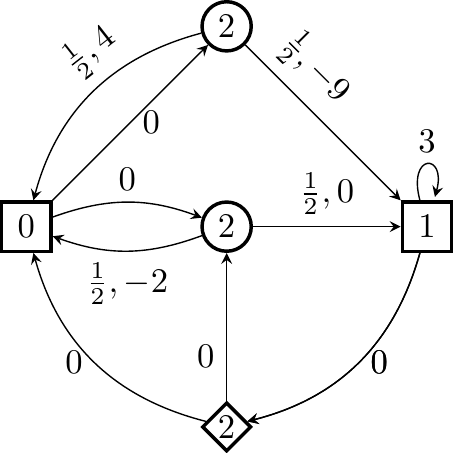}
        \caption{The game $\sys{G}_3$}
        \label{fig:np:c}
    \end{subfigure}
    \caption{\label{fig:np}
        An example game $\sys{G}$ (left) and the derived games.
        The strategy that always loops in the right-most state of $\sys{G}$ ensures a mean-payoff of $3$. As this is the only MD strategy for maximizer that ensures a positive mean-payoff, a factor $f=1$ is sufficient here and we have $\sys{G}_1=\sys{G}$.
        In the derived game $\sys{G}_2$ in \cref{fig:np:b} there are as many trade-in options for the random state as there are MD minimizer's strategies in $\sys{G}_1$ (just two in this example).
        The blue one (top left) corresponds to minimizer going left and the red one (top right) to going up in $\sys{G}_1$.
        Maximizer almost-surely wins $\Gain$ in $\sys{G}$ iff he almost-surely wins a storage-parity condition (see \Cref{thm:storage-parity-comp}) in $\sys{G}_3$.}
\end{figure}

%% file: sec.mainresults.tex
In this section, we prove the main results of the paper, namely that 
almost-sure energy parity stochastic games can be decided in \NP{} and \coNP{}.
The proofs are straightforward and follow from the much more involved characterization of almost sure energy parity objective in terms of the $\Bailout$ and $\Gain$ objectives established in Section \ref{sec:correctness} and their computational complexity analysis in Sections \ref{sec:bailout} and \ref{sec:gain}, respectively.
\begin{theorem}
    \label{thm:main-fixed-k}
    Given an SSG, energy level $k^*$, checking if a state $s$ is almost-sure winning for $\EN{k^*}\cap\Parity$ 
    is in $\NP \cap \coNP$.
\end{theorem}
\begin{proof}
    Recall that we can compute the set $W$ from \cref{def:W} by iterating
   \begin{equation*}
   W_{i} \quad\eqdef\quad\AS{\Gain \cap \always W_{i-1}}\, \cap\,
    \bigcup_{k}\AS{\kBailout \cap \always W_{i-1}}
   \end{equation*}
    starting with $W_0  \eqdef V$,
    until we reach the greatest fixed point $W$.
    Note that at step $i$ we need to solve almost sure $\Gain$ and almost sure $\bigcup_{k}\AS{\kBailout}$,
    where the states of the game are restricted to $W_{i-1}$.
    There can be at most $|V|$ steps, because at least one state is removed in each iteration.
    
    It then suffices to check $\AS{\kBailout \cap \always W_{}}$ (i.e., $\AS{\kBailout}$ for the subgame that consists only of the states of the fixed point $W$ for $k=k^*$. Note that this step can be skipped if $k^* \geq K$, the bound from \cref{thm:bailout-comp}.
    
    Before we discuss how to use $\NP$ and $\coNP$ procedures to construct these sets and to conduct the final test on the fixed point $W$, we note that the `$\cap \always W_{i-1}$' does not add anything substantial, as these are simply the same tests and procedures conducted on the subgame that only consist of the states of $W_{i-1}$.
    
    To obtain an $\NP$ procedure for constructing $\AS{\Gain}$---or, as remarked above, $\AS{\Gain \cap \always W_{i-1}}$---we can guess and validate its membership for each state $s$ \emph{in} this set, using the $\NP$ result from \cref{thm:as-gain-np}, and we can guess and validate its non-membership for each state $s$ \emph{not in} this set in $\NP$, using the $\coNP$ result from Theorem~\ref{thm:coNP-gain}.
    Similarly, we can guess and validate both the membership and the non-membership in $\bigcup_{k}\AS{\kBailout \cap \always W_{i-1}}$---and of $\bigcup_{k}\AS{\kBailout \cap \always W_{i-1}}$ by analysing the subgame with only the states in $W_{i-1}$---by using the $\NP$ and $\coNP$ result, respectively, from \cref{thm:bailout-comp}.
    
    Once we can construct these sets, we can also intersect them and check if a fixed point has been reached. (One can, of course, stop when $s\notin W_i$.)
    
    We can now conduct the final check in $\NP$ using \cref{thm:as-gain-np}.
    
    A $\coNP$ algorithm that constructs $W$ can be designed analogously:
    once $W_{i-1}$ is known, membership and non-membership of a state $s$ in $\AS{\Gain \cap \always W_{i-1}}$ can be guessed and validated in $\coNP$ by \cref{thm:coNP-gain} and by \cref{thm:as-gain-np}, respectively; and membership or non-membership of a state in $\bigcup_{k}\AS{\kBailout \cap \always W_{i-1}}$ can be guessed and validated in $\coNP$ using the $\coNP$ and $\NP$ part, respectively, of \cref{thm:bailout-comp}.
    
    Once $W$ is constructed, we can conduct the final check in $\coNP$ using Theorem~\ref{thm:coNP-gain}.
\end{proof}
This result, together with the upper bound on the energy needed to win
energy-parity objective, allows us to solve the ``unknown initial energy
problem'' \cite{BFLMS2008}, which is to compute the minimal initial energy level required.
\begin{corollary}
For any state $s$, checking if there is $k$ such that $\AS{\EN{k}\cap\Parity}$ holds is in $\NP\cap\coNP$.
Also, for a given $k^*$, checking if $k^*$ is the minimal energy level required to win almost surely is in $\NP\cap\coNP$ as well.
\end{corollary}
\begin{proof}
Due to \cref{thm:bailout-comp}, if there is an energy level $k$ for which $\AS{\EN{k}\cap\Parity}$ holds, then it also holds for the bound $K$ whose size is polynomial in the size of the game. We can then simply calculate $K$ and then use $\NP$ and $\coNP$ algorithms from \cref{thm:main-fixed-k} for $\AS{\EN{K}\cap\Parity}$.

As for the second claim, note that checking whether maximizer cannot win almost surely $\EN{k}\cap\Parity$ is also in $\NP$ and $\coNP$
as a complement of a $\coNP$ and an $\NP$ set, respectively.
Therefore, for an $\NP/\coNP$ upper bound it suffices to simultaneously guess certificates for almost surely $\EN{k^*}\cap\Parity$
and not almost surely $\EN{k^*-1}\cap\Parity$
and verify them in polynomial time.
\end{proof}

%% file: sec.storage-parity.tex
Finally, let us mention that the slightly more restrictive
\emph{storage-parity} objectives can also be solved in $\NP\cap\coNP$.
These are almost identical to energy-parity except that, 
in addition, there must exist some bound $l\in\N$ such that the energy level never drops by more than $l$ during a run.
This extra condition ensures that, if the storage-parity objective holds almost-surely,
then there must exist a \emph{finite-memory} winning strategy for maximizer.

\begin{restatable}{theorem}{thmSTPairityComp}
    \label{thm:storage-parity-comp}
    One can check in $\NP, \coNP$ and pseudo-polynomial time
    if, for a given SSG $\sys{H}\eqdef(V, E, \prob)$,
    $k\in\N$ and control state $s\in V$,
    maximizer can almost-surely satisfy
    $\ES{k}\cap\Parity$
    from $s$.
    
    Moreover, there is a bound
    $L\in\N$, polynomial in the number of states and the largest absolute transition reward,
    so that
    $\ES{k}\cap\Parity = \ES{k,L}\cap\Parity$.
\end{restatable}

\begin{proof}
    (sketch)
This result follows by a simple adaptation of the proofs showing the same computational complexity of the $\Bailout$ objective (\cref{sec:bailout}). See the end of \cref{app:bailout} for further details.    
\end{proof}

\begin{example}
  \label{ex:two-ways-storage}
  In the game in \cref{fig:two-ways}, maximizer cannot ensure the storage-parity condition $\ES{k}\cap\Parity$
  for any initial energy level $k$. This is because it would imply the existence of a finite-memory almost-surely winning strategy, which as we have already argued, cannot be true.
  More intuitively, 
  to prevent an intermediate energy drop by $l$ units, a winning maximizer strategy for
  storage-parity would need to stop moving left after observing the negative
  cycle in the leftmost state $l$ successive times. However, when maximizer
  moves right, this gives minimizer the chance to visit the rightmost bad
  state (with dominating odd priority $1$).
  The chance of that happening is $(1/3)^l > 0$. In particular,
  this probability is $>0$ for any value of the
  intermediate energy drop $l$. Therefore, for any fixed $l$, maximizer would need to move right
  infinitely often to satisfy storage and lose (against an optimal minimizer strategy that moves
  to the rightmost state).
\end{example}

%% file: sec.conclusion.tex
We showed that several almost-sure problems for combined energy-parity
objectives in simple stochastic games are in $\NP \cap \coNP$. 
No pseudo-polynomial algorithm is known (just like for stochastic mean-payoff parity games \cite{CDGQ:2014}).
All these problems subsume (stochastic) parity games, by setting all rewards to $0$. 
Thus the existence of a pseudo-polynomial algorithm would imply that
(stochastic and non-stochastic) parity games are in P, which is a long-standing open problem.

It is known that maximizer already needs infinite memory to win almost-surely a combined energy-parity objective in MDPs \cite{MSTW2017}. Our results do not imply anything about the memory requirement for optimal minimizer strategies in SSGs for this objective.
We conjecture that memoryless minimizer strategies suffice. If this conjecture holds (and is proven), this would greatly simplify the $\coNP$ upper bound that we established for this problem.

A natural question is whether results on
mean-payoff/energy/parity games can be generalized
to a setting with multi-dimensional payoffs.
Non-stochastic multi-mean-payoff and multi-energy
games have been studied in \cite{VCDHRR2015,JLS2015,AMSS2013}.
To the best of our knowledge, the techniques used there,
e.g.\ upper bounds on the necessary energy levels as in \cite{JLS2015}, do not
generalize to stochastic games (or MDPs).

Multiple mean-payoff objectives in MDPs have been studied in
\cite{Kuceraetal2014,KKK2027},
but the corresponding multi-energy (resp.\ multi-energy-parity) objective has extra difficulties due
to the 0-boundary condition on the energy.
I.e., even on Markov chains, and without any parity condition, it subsumes
problems about multi-dimensional random walks.
Some partial results on Markov chains and MDPs
have been obtained in \cite{14BKKNK-LICS,ACMSS2016,lmcs:867},
but the decidability of the almost-sure problem for stochastic multi-energy-parity games (and MDPs) remains open.

%% file: app.lifting.tex
This section contains a proof of \cref{lem:FD-strategy},
that shows how to conclude lift the existence of memoryless determined
almost-sure winning strategies from MDPs to SSGs.

\begin{definition}[\cite{GK2014}, Sec.\ 2.C]\label{def:eps-subgame-perfect}
Let $\?G=(V,E,\prob)$ be an SSG and $\sigma$ a maximizer strategy.

For every finite play $p = s_0 \dots s_n$
we denote by $\sigma[p]$
the shift of strategy $\sigma$ by $p$ as the strategy defined by
\[
\sigma[p](t_0 e_1 t_1 \dots e_m t_m) \eqdef
\begin{cases}
\sigma(p e_1 t_1 \dots e_m t_m), & \text{if } s_n=t_0\\
\sigma(t_0 e_1 t_1 \dots e_m t_m), & \text{otherwise}
\end{cases}
\]
Then $\sigma$ is said to be $\eps$-subgame-perfect if for every finite play
$p$ the strategy $\sigma[p]$ is $\eps$-optimal.
\end{definition}

\begin{definition}[\cite{GK2014}, Sec.\ 5.C]\label{def:trigger-strategy}
Let $\?G=(V,E,\prob)$ be an SSG from initial state $s$ 
and $\?O$ an objective
that is both shift-invariant and submixing and $\pi \in V$ a 
state owned by maximizer.
We assume w.l.o.g., that $\pi$ has two successors
left ($l$) and right ($r$).
Let $\?G_l$ and $\?G_r$ be the SSGs resulting from $\?G$ by removing the edge
from $\pi$ to $r$ and $l$, respectively.
Moreover, assume that $\tau_l$ and $\tau_r$ are $\epsilon$-subgame-perfect
strategies for minimizer in $\?G_l$ and $\?G_r$, respectively.

The \emph{trigger strategy} $\tau_{lr}$ 
for minimizer in the original game $\?G$ (starting at $s$) is defined as follows:
    \begin{itemize}
        \item start by playing according to $\tau_l$
        \item play according to $\tau_r$ (initially with empty memory)
            once maximizer moves from $\pi$ to $r$ for the first time.
        \item every time maximizer moves from $\pi$ to $l$ (or to $r$), maximizer
            resumes the previous play in $\?G_l$ (or $\?G_r$).
    \end{itemize}
    The trigger strategy $\tau_{lr}$ allocates
    the memory used by $\tau_l$, $\tau_r$, and one extra bit to remember
    maximizer's last choice at $\pi$.
\end{definition}

\begin{lemma}[\cite{GK2014}, Eq. (19),(20),(21)]\label{lem:gimbert-kelmendi}
Assume the definitions of \cref{def:trigger-strategy}.
Then
\[
\forall \sigma.\, \Prob[\sigma,\tau_{lr}]{s}(\?O) \le \max\{\Val[\sys{G}_l]{s}{\?O}, \Val[\sys{G}_r]{s}{\?O}\}
    + \epsilon
\]
\end{lemma}

\lemFDstrategy*
\begin{proof}
    We assume w.l.o.g., that all minimizer's states have at most two successors.

    The proof is done by induction on the number of minimizer's states with choice
    (two successors) in $\?G$.
    The base case holds by the assumption that maximizer has FD optimal strategies in MDPs.
    
    For the induction step, we will use \cref{def:trigger-strategy}
    and \cref{lem:gimbert-kelmendi}, instantiated with $\overline{\?O}$
    instead of $\?O$. Since $\overline{\?O}$ is both shift-invariant and
    submixing, this satisfies the conditions of \cref{def:trigger-strategy}, 
    but (relative to $\?O$) the roles of the players minimizer/maximizer are swapped.

    Pick some initial state $s$ and a minimizer's state $\pi$ for $\?O$ (i.e., a maximizer's state for $\overline{\?O}$)
    and let $\?G_l$, $\?G_r$ be defined as in \cref{def:trigger-strategy}.
    By induction hypothesis, in both these games $\?G_l$ and $\?G_r$, maximizer has an FD optimal
    strategy for objective $\?O$ from $s$. Call these strategies $\sigma_l$
    and $\sigma_r$, respectively. 
    In particular, since $\sigma_l$ and $\sigma_r$ are optimal and 
    $\?O$ and $\overline{\?O}$ are shift-invariant, the strategies $\sigma_l$ and $\sigma_r$
    are subgame-perfect, and thus $\epsilon$-subgame-perfect for $\epsilon=0$.
    Thus we can instantiate \cref{def:trigger-strategy} with objective
    $\overline{\?O}$ and reversed roles of players minimizer/maximizer. I.e.,
    we take $\sigma_l$ for $\tau_l$ and $\sigma_r$ for $\tau_r$, which are
    subgame-perfect for player minimizer for objective $\overline{\?O}$.
    We obtain the trigger-strategy $\sigma_{lr}$ for maximizer for
    $\?O$ (i.e., the $\tau_{lr}$ for minimizer for $\overline{\?O}$ from \cref{def:trigger-strategy}).
    Since $\sigma_l$ and $\sigma_r$ are FD, so is $\sigma_{lr}$.
 
    We now argue that this trigger strategy $\sigma_{lr}$ must be optimal.
    The shift-invariance and submixing conditions on $\overline{\?O}$ imply
    (\cite{GK2014}, Theorem 5.2) that minimizer has MD optimal strategies in every SSG with winning condition $\?O$.
    Let $\tau^*$ be some MD optimal strategy for minimizer in $\?G$ from $s$.
    W.l.o.g.\ assume that $\tau^*(\pi)=l$ (otherwise rename left/right).

    We show that $\sigma_{lr}$ and $\tau^*$ are best responses to each other,
    and thus both are optimal.
    That is, in order to finish the induction step, we prove that
    the following two claims hold for the game $\?G$.
    \begin{enumerate}
        \item
        $\Prob[\sigma_{lr},\tau^*]{s}(\?O) \ge \sup_{\sigma}\Prob[\sigma,\tau^*]{s}(\?O)$,
        and \label{lem:FD-strategy-ad1}
        \item
        $\Prob[\sigma_{lr},\tau^*]{s}(\?O) \le \inf_{\tau}\Prob[\sigma_{lr},\tau]{s}(\?O)$.
        \label{lem:FD-strategy-ad2}
    \end{enumerate}
    Together these imply the claim that $\sigma_{lr}$ is optimal, and hence the induction step, because
    \begin{align}
    \Val{s}{\?O}
    = \sup_\sigma \inf_\tau\Prob[\sigma,\tau]{s}(\?O)
    \eqby{(opt.)} \sup_\sigma \Prob[\sigma,\tau^*]{s}(\?O)
    \eqby{(1)} \Prob[\sigma_{lr},\tau^*]{s}(\?O)
    \eqby{(2)} \inf_{\tau}\Prob[\sigma_{lr},\tau]{s}(\?O)
    \end{align}
    where the second equation uses the optimality of $\tau^*$.
    It remains to prove the two claims above.

    \cref{lem:FD-strategy-ad1}).
    Since $\tau^*(\pi)=l$ we have
    $\Prob[\sigma_{lr},\tau^*]{\?G, s}(\?O)
    = \Prob[\sigma_{l},\tau^*]{\?G_l, s}(\?O)
    \ge \sup_{\sigma}\Prob[\sigma,\tau^*]{\?G_l, s}(\?O)
    = \sup_{\sigma}\Prob[\sigma,\tau^*]{\?G, s}(\?O)$,
    where the equalities hold by $\tau^*(\pi)=l$
    and the inequality holds by the assumed optimality of $\sigma_l$ in $\?G_l$.
    
    \cref{lem:FD-strategy-ad2}).
    From \cref{lem:gimbert-kelmendi}, instantiated with $\overline{\?O}$, we
    obtain that
    \[
    \forall \tau.\, \Prob[\sigma_{lr},\tau]{s}(\overline{\?O}) \le \max\{\Val[\sys{G}_l]{s}{\overline{\?O}}, \Val[\sys{G}_r]{s}{\overline{\?O}}\}.
    + \epsilon
    \]
    Since in our case $\epsilon=0$ we obtain
    \[
    \forall \tau.\,
    1-\Prob[\sigma_{lr},\tau]{s}({\?O}) \le \max\{1-\Val[\sys{G}_l]{s}{\?O},
    1-\Val[\sys{G}_r]{s}{\?O}\}
    = 
    1 - \min\{\Val[\sys{G}_l]{s}{\?O}, \Val[\sys{G}_r]{s}{\?O}\}
    \]
    and thus
    \[
    \forall \tau.\, \Prob[\sigma_{lr},\tau]{s}({\?O}) \ge \min\{\Val[\sys{G}_l]{s}{\?O}, \Val[\sys{G}_r]{s}{\?O}\}.
    \]
    In particular, for $\tau = \tau^*$ we obtain
    $
    \Prob[\sigma_{lr},\tau^*]{s}({\?O}) \ge \min\{\Val[\sys{G}_l]{s}{\?O}, \Val[\sys{G}_r]{s}{\?O}\}
    $.
    However, since $\tau^*$ is an MD optimal strategy for minimizer, we also have
    \[
    \Prob[\sigma_{lr},\tau^*]{s}({\?O}) \le \min\{\Val[\sys{G}_l]{s}{\?O}, \Val[\sys{G}_r]{s}{\?O}\}
    \]
    By combining the above we get
    \[
    \forall \tau.\, \Prob[\sigma_{lr},\tau]{s}({\?O}) \ge \min\{\Val[\sys{G}_l]{s}{\?O}, \Val[\sys{G}_r]{s}{\?O}\}
    = 
    \Prob[\sigma_{lr},\tau^*]{s}({\?O})
    \]  
    and thus 
    $\inf_{\tau}\Prob[\sigma_{lr},\tau]{s}({\?O}) \ge \Prob[\sigma_{lr},\tau^*]{s}({\?O})$.
    This concludes the proof of \cref{lem:FD-strategy-ad2} and thus the induction step.
\end{proof}

%% file: app.4.tex
We will proceed in several reduction steps, ultimately reducing to checking the winner of a non-stochastic game for energy-parity objectives.

Assume from now on a fixed SSG $\sys{G}$ with associated reward and parity functions.

\begin{lemma}
\label{lem:bailout'}
Let $\Bailout'(k,l)\eqdef (\ES{k,l}\cap\Parity)\cup(\ES{k,l}\cap\LimSup{=\infty})$.

There exists $L\in\N$ so that
$\AS{\bigcup_{l}\Bailout(k,l)}
        =
    \AS{\Bailout'(k,L)}$.
\end{lemma}
\begin{proof}
    Pick $L$ larger than $\card{V}\cdot R\cdot c$, the number of control states in the game times the largest absolute reward $R$ times the largest priority $c$ used in the parity condition.

We claim that every a.s.~winning strategy can be turned into one that avoids sub-runs of the form
$s\step{\pi_1}s\step{\pi_2}s$ where
1) both $\pi_1$ and $\pi_2$ have strictly negative total effect on the energy level,
2) neither $\pi_1$ nor $\pi_2$ visit state $s$ internally,
3) the dominant priority on $\pi_1$ and $\pi_2$ is the same.
If a strategy allows such a path, then one can safely ``cut out'' $\pi_2$ and the resulting strategy will still be a.s.~winning.
Taken to the limit, such transformations will result in a strategy that is a.s.~winning for
$\AS{\Bailout'(k,L)}$.
\end{proof}

\begin{lemma}
\label{lem:bailout''}
Let $\Bailout''(k,l)\eqdef (\ES{k,l}\cap\Parity)\cup(\ES{k,l}\cap\LimInf{=\infty})$.

For every $k,l\in\N$ it holds that
$\Bailout'(k,l) = \Bailout''(k,l)$.
\end{lemma}
\begin{proof}
    Just notice that a run $\rho=s_0e_0s_1e_1\ldots \in
    \ES{k,l}\cap\LimSup{=\infty}$ must also satisfy the
    $\LimInf{=\infty}$ condition because
$(\lim\inf_{n\to\infty} \sum_{i=0}^{n}\cost(e_i)) \ge
(\lim\sup_{n\to\infty} \sum_{i=0}^{n}\cost(e_i)) - l$,
by the $l$-storage assumption.
\end{proof}

The idea of the next step
is to allow maximizer to witness the $\LimInf{=\infty}$ condition by occasionally trading in energy for a good priority, thereby satisfying a parity condition instead.
This results in a stochastic game for a $\ES{k,l}\cap\Parity{}$ objective.

Let $\sys{G}'$ be the SSG derived from $\sys{G}$, where
maximizer can always trade energy-increase for visiting the best possible priority $0$.
That is, $\sys{G}'$ results from $\sys{G}$ by replacing every edge $s\step{+a}t$, with $a>0$, by a gadget
below, where $s'\in \VMax$,
$\parity(s')=\parity(s)$ and
$\parity(t')=0$.

\begin{center}
\begin{tikzpicture}[node distance=1cm and 2cm]
    \node (s) {$s$};
    \node[right=of s] (s') {s'};
    \node[above right=of s'] (t') {t'};
    \node[below=of t'] (t) {$t$};
    \draw[->] (s)  edge node[above] {$0$} (s');
    \draw[->] (s') edge node[auto] {$0$} (t');
    \draw[->] (t') edge node[auto] {$0$} (t);
    \draw[->] (s') edge node[above] {$+a$} (t);
\end{tikzpicture}
\end{center}

\begin{lemma}
\label{lem:bailout'''}
    For every state $s$ of $\sys{G}$, and every $k,l\in\N$ it holds that
    $s\in \AS[\sys{G}]{\Bailout''(k,l)}$ if, and only if, $s\in \AS[\sys{G}']{\ES{k,l}\cap\Parity}$.
\end{lemma}
\begin{proof}
    Assume that $R\in\N$ is the largest absolute transition reward in $\sys{G}$ (and hence also $\sys{G}'$).
    Every a.s.~winning strategy $\sigma$ for $\Bailout''(k,l)=\ES{k,l}\cap(\Parity\cup\LimInf{=\infty})$ in $\sys{G}$
    can be turned into an a.s.~winning strategy $\sigma'$ for $\ES{k,l}\cap\Parity$ in $\sys{G}'$ as follows.

    The new strategy $\sigma'$ behaves just as $\sigma$ but additionally, keeps track of the energy levels up to the bound $l\cdot R$. If in $\sys{G}$, $\sigma$ chooses to increase the energy level above this bound, $\sigma'$ will opt to visit a good priority instead, and continue from the current energy level.
    Since $\sigma$ ensures the $l$-storage condition
    on (almost) all runs, so does $\sigma'$.
    Moreover, plays in $\sys{G}$ that do not satisfy
    $\Parity$ must instead satisfy $\LimInf{=\infty}$.
    The corresponding runs in $\sys{G}'$ according to $\sigma'$ will therefore infinitely often visit the best priority
    and hence satisfy the parity condition.

    For the other direction, notice that one can just as well transform an a.s.~winning strategy $\sigma'$ for storage-parity in $\sys{G}'$ to a winning strategy $\sigma$ for
    $\Bailout''(k,l)$ in $\sys{G}$.
    The strategy $\sigma$ just increments the energy level and whenever $\sigma'$ would visit a newly introduced priority-$0$ state.
    Suppose $\rho$ is a play in $\sys{G}$ that corresponds to a play $\rho'$ in $\sys{G}'$.
    If $\rho'$ visits new states only finitely often, then after some finite prefix, the sequence of states visited by $\rho'$ and $\rho'$ are the same. Since $\rho'$ satisfies the parity condition so must $\rho$.
    Otherwise, if $\rho'$ visits new states infinitely often, then $\rho$ the difference of energy levels
    on $\rho$ and $\rho'$ must grow unboundedly. Since $\rho'$ satisfies the $l$-storage condition
    this means that $\rho$ satisfies the $\LimInf{=\infty}$ condition, and hence $\Bailout''(k,l)$.
\end{proof}
    
Finally, we use a construction similar to that in \cite{CJH2004}
for parity objectives, to replace random states by small ``negotiation gadgets'',
resulting in a non-stochastic energy-parity game.
Let $\sys{G}''$ be the non-stochastic game derived from $\sys{G}'$, where random states are replaced by
gadgets as in \cite{CJH2004}.

\begin{lemma}
\label{lem:bailout:2pg}
    For every state $s$ of $\sys{G}'$ and every $k,l\in\N$ it holds that
    $s\in \AS[\sys{G}']{\ES{k,l}\cap\Parity}$ if, and only if, $s\in \AS[\sys{G}'']{\ES{k,l}\cap\Parity}$.
\end{lemma}
\begin{proof}
The construction in \cite{CJH2004} does not affect the transition rewards.
Thus the $\ES{k,l}$ condition is trivially preserved.
The a.s.\ $\Parity$ condition is preserved by exactly the same argument as
in \cite{CJH2004}.
\end{proof}

\thmBailoutComp*
\begin{proof}
By \cref{lem:bailout',lem:bailout'',lem:bailout''',lem:bailout:2pg},
for every $k\in\N$ it holds that
\begin{align*}
\AS[\sys{G}]{\bigcup_{l}\Bailout(k,l)}
&\eqby{(L.~\ref{lem:bailout'})}   \AS[\sys{G}]{\Bailout'(k,L)}\\
&\eqby{(L.~\ref{lem:bailout''})}  \AS[\sys{G}]{\Bailout''(k,L)}\\
&\eqby{(L.~\ref{lem:bailout'''})} \AS[\sys{G}']{\ES{k,L}\cap\Parity}\\
&\eqby{(L.~\ref{lem:bailout:2pg})}\AS[\sys{G}'']{\ES{k,L}\cap\Parity}
\end{align*}
Since $\sys{G}''$ is a two-player non-stochastic game, the claim now follows from \cite{CD2012}, (Theorem~2 and Lemma~5). For the existence of polynomially bounded number $K,L$ just notice that $\sys{G}''$
has the same largest absolute transition reward, and only a polynomially larger set of states compared to $\sys{G}$.
For non-stochastic energy-parity games such as $\sys{G}''$
it holds that
$\bigcup_{k\ge 0}\AS{\ES{k,L}\cap\Parity}
=\AS{\bigcup_{k\ge 0}(\ES{k,k}\cap\Parity)}
=\AS{\EN{K}\cap\Parity}
$, if $K$ denotes the product of the number of states, the largest priority and absolute transition rewards in $\sys{G}''$.

Now, to check if a state $s$ belongs to $\bigcup_{k\ge 0}\AS[\sys{G}]{\kBailout}$, we can calculate $K$ and then simply follow the $\NP$ or $\coNP$ procedure to check if $s$ belongs to $\AS[\sys{G}]{\text{K-}\Bailout}$ instead.
This shows that this problem in $\NP$ and $\coNP$ as well.
\end{proof}

As a side result, note that neither Lemma \ref{lem:bailout:2pg}, nor the complexity argument in Theorem \ref{thm:bailout-comp},
make use of the structure of $\sys{G}'$:
they hold for all SSGs with storage parity condition.

\thmSTPairityComp*

%% file: app.gain.tex
\subsection{Strategy Complexity for $\Gain$} 

We prove
\cref{lem:MDP-FD-strategy}, i.e.,
if maximizer can almost-surely win $\Gain$ in an MDP, then he can do so using a finite-memory deterministic strategy.

To do this, we will utilize some results from \cite{MSTW2017},
where we showed how to compute winning regions for energy-parity objectives 
in MDPs based on a similar combination of ``gain'' and ``bailout'' objectives as in this paper.

Consider a state $s$ of a finite MDP with energy-parity objective
and define the \emph{limit value} of state $s$ as
$\LVal{s} \eqdef \sup_{k}\Val{s}{\EN{k}\cap\Parity}$.
This is well defined, because energy conditions are monotone increasing in the initial energy level $k$.

\begin{lemma}
    \label{lem:MDP-LVAL}
 For any state $s$ of a finite MDP, we have $\Val{s}{\Gain}=\LVal{s}$.
\end{lemma}
\begin{proof}%
It follows directly from the definitions that for every $k \in \N$
\[
\EN{k}\cap\Parity \subseteq \LimInf{\ge -k}\cap\Parity \subseteq \LimInf{> -\infty}\cap\Parity = \Gain
\]
and thus
\begin{equation}\label{lemMDPFD-1}
\bigcup_k(\EN{k}\cap\Parity) \subseteq \Gain
\end{equation}
Towards the reverse inclusion, consider a run
$\rho \in \LimInf{\ge -j}\cap\Parity$ for some $j \in \N$.
Then, except in a finite prefix $\rho'$ , the energy along $\rho$ stays above
$-j$. Let $k'$ be the minimal energy reached in $\rho'$, which is finite
because $\rho'$ is finite, and let $k \eqdef -\min(k', -j)$.
Then $\rho \in \EN{k}\cap\Parity \subseteq \bigcup_k(\EN{k}\cap\Parity)$.
So for every $j \in \N$ we have
\[
\LimInf{\ge -j}\cap\Parity \subseteq \bigcup_k(\EN{k}\cap\Parity)
\]
and thus
\begin{equation}\label{lemMDPFD-2}
\Gain = \bigcup_j(\LimInf{\ge j}\cap\Parity) \subseteq \bigcup_k(\EN{k}\cap\Parity)
\end{equation}
From \cref{lemMDPFD-1} and \cref{lemMDPFD-2} we obtain
\begin{equation}\label{lemMDPFD-3}
\Gain= \bigcup_k(\EN{k}\cap\Parity)
\end{equation}
    Therefore,
    \begin{align*}
    \Val{s}{\Gain}
    &=\Val{s}{\cup_k(\EN{k}\cap\Parity)} & \text{by \cref{lemMDPFD-3}}\\
    &=\sup_\sigma\Prob[\sigma]{s}(\cup_k(\EN{k}\cap\Parity)) & \text{def.\ of value}\\
    &=\sup_\sigma\sup_k\Prob[\sigma]{s}(\EN{k}\cap\Parity)
    & \text{continuity of measures from below}\\
    &=\sup_k\sup_\sigma\Prob[\sigma]{s}(\EN{k}\cap\Parity) & \text{commutativity}\\
    &=\sup_k\Val{s}{\EN{k}\cap\Parity} & \text{def.\ of value}\\
    &=\LVal{s} & \text{def.\ of $\LVal{s}$}
    \end{align*}
\end{proof}

\lemMDPFDstrategy*

\begin{proof}
  By \cref{lem:MDP-LVAL} we have $\Val{s}{\Gain} = \LVal{s}$.
  Moreover, the objective $\Gain$ is shift-invariant and therefore there exist optimal strategies \cite{GH2010}.
  Thus it follows from \cite[Theorem~18]{MSTW2017} that
    $\AS{\Gain} = \AS{\reachset{A\cup B}}$, for the following sets of states $A\eqdef\bigcup_{k\in\N}\AS{\ES{k}\cap\Parity}$
    and $B\eqdef\AS{\LimInf{=\infty}\cap\Parity}$.
    This means that if an a.s.~winning strategy for $\Gain$ exists, then there
    also exists one that operates in two phases:
    1) a.s.~reach $A\cup B$. This can be done with memoryless deterministic strategies.
    2a) once in $A$ proceed along an a.s.~winning strategy for $\ES{k}\cap\Parity$,
    which can be done deterministically with memory $\?O(k\cdot \abs{G})$. Or,
    2b) once in $B$, proceed along an a.s.~winning strategy for $\LimInf{=\infty}\cap\Parity$.
    For MDPs a strategy is almost-sure winning for $\LimInf{=\infty}\cap\Parity$
    iff it is almost-sure winning for $\MP{>0}\cap\Parity$,
    the combination of a parity condition together with a \emph{strictly} positive Mean-Payoff condition.
    Such strategies can be chosen FD \cite{CD2011}.
\end{proof}

\subsection{$\Gain$ is in $\coNP$} 
\label{app:gain-coNP}
\input{app.gain-coNP}

\subsection{$\Gain$ is in $\NP$} 
\label{app:gain-NP}
\input{app.gain-NP}
\input{ex.NP}

%% file: app.gain-coNP.tex
\coNPgain*
\begin{proof}
By \cref{lem:gain:min-MD}, it suffices to show \coNP\ membership only for the MDP case,
as a witnessing MD strategy for minimizer can be guessed as part of the certificate. 
To check if maximizer can almost surely win from state $s$ in an MDP with $\Gain$ objective,
we can equivalently check if $\Val{s}{\Gain}=1$.
This is because 
the objective is shift-invariant and therefore there exist optimal strategies \cite{GH2010}.
By \cref{lem:MDP-LVAL}, we can alternatively check if $\LVal{s}=1$, which can be done in $\coNP$ by \cite[Lemma~26]{MSTW2017}.
\end{proof}

%% file: app.gain-NP.tex
Before we can proceed with the technical details of the
$\sys{G} \to \sys{G}_1 \to \sys{G}_2 \to \sys{G}_3$ constructions,
we first need to introduce the following standard definitions.

\begin{definition}
    Let $\sys{M}=\sys{G}[\tau]$ be an MDP induced by game $\sys{G}\eqdef(V=(\VMax,\VMin,\VRan), E, \prob)$ and an MD strategy $\tau$ for minimizer.
    An \emph{end-component} is a strongly connected set of states $C \subseteq V$ such that, for every state $v\in C$,
    if $v \in \VMax$ then some successor $v'$ of $v$ is in $C$, and
    if $v \in \VRan\cup\VMin$ then all successors $v'$ of $v$ are in $C$.
    A \emph{leaf-component} is an end-component of a Markov chain $\sys{G}[\sigma,\tau]$.
    
    A leaf-component is \emph{\safe} if the dominating priority is even and it satisfies the storage condition $\bigcup_{k\ge 0}{\ES{k}}$, and 
    \emph{\positive} if its mean-payoff is positive.
    
    An end-component $C$ of $\sys{G}[\tau]$ is \emph{\awesome} if (1) the dominating priority of $C$
    (the smallest priority of any state in $C$) is even and $C$ contains a \positive leaf-component
    \emph{or} (2)
    there is an MD strategy $\sigma$ for the maximizer, such that $C$ is a \safe leaf-component in $\sys{G}[\sigma,\tau]$. 
\end{definition}

Note that any end-component that satisfies a.s.~$\Gain$ is \awesome, which justifies its name. This is because either (1) holds, or else maximizer can reach again and again a state with a dominating even priority without the need to pump up the energy level first for which an MD strategy suffices, so (2) would hold then.

\subsubsection{Blow-up Construction ($\sys{G} \to \sys{G}_1)$}
\rule[-10pt]{0pt}{10pt}\\
\label{app:g-to-g1}
\input{sec.gain.NP.blowup}
\input{app.6.NP.blowup}

\subsubsection{Trade-in Construction ($\sys{G}_1 \to \sys{G}_2)$}
\rule[-10pt]{0pt}{10pt}\\
\label{app:g1-to-g2}
\label{sec:gain.NP.tradein}
\input{sec.gain.NP.tradein}

\input{app.6.NP.tradein}

\subsubsection{Concise Witnesses Construction ($\sys{G}_2 \to \sys{G}_3$)}
\label{app:g2-to-g3}
\rule[-10pt]{0pt}{10pt}\\
\input{sec.gain.NP.concise}

\input{ex.NP}

%% file: sec.gain.NP.blowup.tex
The $\sys{G} \to \sys{G}_1$ construction just multiplies all the rewards by a ``large enough`` factor.
Formally, we need $\sys{G}_1$ to have the following property.
\begin{restatable}{lemma}{meanpayoffLemma}
    \label{lem:mean-payoff-2}
    Let $\tau$ be an MD strategy for the minimizer and $s\in V$. %
    If there exists a strategy $\sigma$ for the maximizer such that
    $\Prob[\sys{G},\sigma,\tau]{s}(\MP{>0} \cap \Parity) = 1$
    then there exists an FD strategy $\sigma'$ for the maximizer such that
    $\Prob[\sys{G}_1,\sigma',\tau]{s}(\MP{> 2} \cap \Parity) = 1$.
\end{restatable}

%% file: app.6.NP.blowup.tex
We construct a game $\sys{G}_1$, based on $\sys{G}$, in which all edge rewards are multiplied by a large factor
so that if the maximizer can originally ensure the parity condition and a positive expected mean-payoff in $\sys{G}$, then he can ensure parity condition and expected mean-payoff higher than 2 in $\sys{G}_1$. It is intuitively clear that such a factor exists, because multiplying all transition rewards by a positive factor has no effect on the outcome of the $\Gain$ objective. What is less clear that such a factor can be of polynomial size so that $\sys{G}_1$ is only polynomially larger than $\sys{G}$.
Before we can proceed with the proof of \cref{lem:mean-payoff-2}, we need to show an auxiliary result below.

Recalling that, for the $\Gain$ objective, the minimizer has MD optimal strategies,
we consider the effect of multiplying the rewards of all edges by factor $f$ against every such strategy $\tau$:
We show that if maximizer can a.s.~obtain $\MP{>0}\cap\Parity$ from a state $s$ in the MDP $\sys{G}[\tau]$,
then he can a.s.~obtain $\MP{>2}\cap\Parity$ from $s$ in the MPD $\sys{G}_1[\tau]$.
\begin{lemma}
    \label{lem:np-posMP}
    Let (1) $\tau$ be an MD strategy for minimizer,
    (2) $E$ be an end component in the MDP $\sys{G}[\tau]$, with even minimal priority,
    (3) $\sigma$ an MD strategy for Max, and
    (4) $L \subseteq E$ a leaf component in $\sys{G}[\sigma,\tau]$ with expected payoff $p>0$.
    Then $\frac{2}{p}$ is exponential in the size of $\sys{G}$. %
    
    Moreover, a factor $f>\frac{2}{p}$, with a representation polynomial in $|\sys{G}|$, can be computed independent of $\tau$, E, $\sigma$, or $L$.
\end{lemma}
\begin{proof}
    For any fixed MD strategies $\sigma$ and $\tau$, we can write a linear program for the so-called {\em gain-bias relations}%
    \footnote{The terms `gain' from the `gain-bias relations' and our `$\Gain$ objective' are unrelated established terms.} in $L$, which
    is a standard way to solve MDPs with a mean-payoff objective (see, e.g., \cite[Theorem 8.2.6(a), p. 343]{P1994}).
    In any solution, the gain of a state equals its mean-payoff value while, broadly speaking,
    the bias compensates for the fluctuation of the payoff, where the gain is only the expected longterm average.
    
    Notice that for a fixed $L$, we only need a single gain variable $g$, because all nodes in a leaf component have the same mean-payoff. 
    For each node $u \in L$, we introduce  
    a bias variable, $b_u$. 
    
    The constraints of the gain-bias linear program for $L$ are:
    \begin{align*}
    b_u &= b_{\tau(u)} + \cost(u,\tau(u)) - g &\text{ for all } u \in \VMin \cap L\\
    b_u &= b_{\sigma(u)} + \cost(u,\sigma(u)) - g &\text{ for all } u \in \VMax \cap L\\
    b_u &= \sum_{(u,v) \in E} \prob(u,v) (b_v + \cost(u,v)) - g &\text{ for all } u \in \VRan \cap L
    \end{align*}
    and its objective is {\it Maximize} $g$.
    
    It follows from the proof of Corollory 10.2a in \cite{schrijver1998theory} that the size of an optimal finite solution to such this linear program is at most $4m^2(m+1)(S+1)$, where $m$ is the number of variables and $S$ is the maximum size of any coefficient used. In our case we can easily estimate that $m \leq |V|+1$ and
    $S \leq |\sys{G}|$, so the optimal solution, $p$, is of size polynomial in $|\sys{G}|$. And, since $p > 0$, the same holds for $2/p$.
    
    Note that the loose upper bound given above on the size of $2/p$ does not really depend on $\tau$, $\sigma$, $E$ nor $L$, so if we take the maximum of the size of $2/p$ over all possible $\tau$, $\sigma$, $E$ and $L$, we would still get the same upper bound.
\end{proof}

Such an $f$ will serve as our sufficiently large (yet sufficiently small) blow-up factor: $\sys{G}_1$ is obtained from $\sys{G}$ by changing the reward function to $\cost_1(e) = f \cdot \cost(e)$ for all $e \in E$, i.e., by multiplying all rewards by $f$.
We are now finally ready to prove \cref{lem:mean-payoff-2}.

\begin{proof}[Proof of \cref{lem:mean-payoff-2}]
    The existence of an FD strategy $\sigma'$ that achieves $\Prob[\sys{G},\sigma',\tau]{s}(\MP{>0} \cap \Parity) = 1$ follows from \cite{CD2011}. Moreover, $\sigma$ achieves the same mean-payoff, denoted by $p'$, as the original almost-sure winning strategy $\sigma$.
    By \cref{lem:np-posMP}, 
    the mean payoff of $\sigma'$ in $\sys{G}_1$ is $\ge f \cdot p' > 2$.
\end{proof}

\begin{figure}
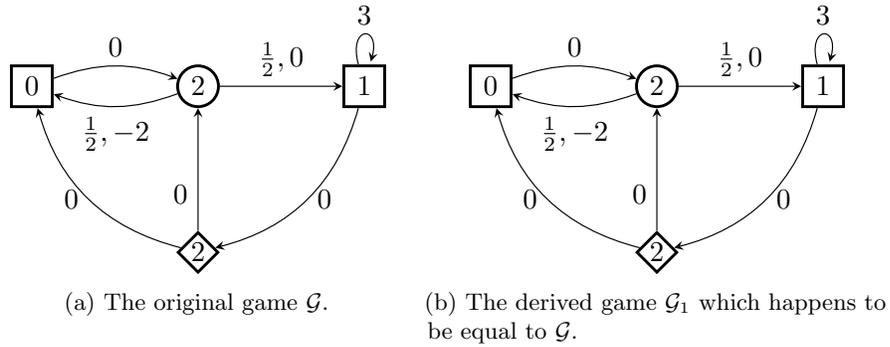

    \begin{subfigure}[t]{0.5\textwidth}
        \centering
        \includegraphics[width=5cm]{Figures/np-example-G.pdf}
        \caption{The original game $\sys{G}$.}
        \label{fig:np1:a}
    \end{subfigure}\hfill
    \begin{subfigure}[t]{0.5\textwidth}
        \centering
        \includegraphics[width=5cm]{Figures/np-example-G.pdf}
        \caption{The derived game $\sys{G}_1$ which happens to be equal to $\sys{G}$.}
        \label{fig:np1:b}
    \end{subfigure}
    \caption{An example game $\sys{G}$ (left) and its example derived game $\sys{G}_1$.
        \label{fig:np1}
    }
\end{figure}

\begin{example}[running example]
    \label{ex:gainNPExample}
    Consider the game $\sys{G}$ in \cref{fig:np1} (left). Maximizer can almost-surely guarantee the $\Gain$ condition.
    The strategy that always loops in the right-most state ensures a mean-payoff of $3$. As this is the only MD strategy for maximizer that ensures a positive mean-payoff, picking any factor $f>\frac{2}{3}$ is  sufficient. In particular we can pick $f=1$ which results in $\sys{G}_1=\sys{G}$.
\end{example}

%% file: sec.gain.NP.tradein.tex
We are now going to modify the game $\sys{G}_1$ into the game $\sys{G}_2$,
where maximizer can sacrifice part of the reward he would normally get while visiting a probabilistic node in exchange for rebalancing the values of these rewards.

During the construction of 
$\sys{G}_2$ we are going to fix an
optimal MD strategy, $\tau^*$, for minimizer in 
$\sys{G}_1$.
Game $\sys{G}_2$ will be the same no matter which optimal strategy is picked as $\tau^*$.

We start the construction of $\sys{G}_2$ with identifying the union, $U$, of all 
\awesome\ end-components of $\sys{G}_1[\tau^*]$, for which there is no maximizer strategy that ensures $\MP{>0} \cap \Parity$.
Condition (2) of \awesome{}ness has to hold instead, i.e., there are MD maximizer strategies that a.s.\ satisfy storage and parity, and note that then the mean-payoff has to be 0.
We can compose all these strategies into 
a single winning maximizer MD strategy $\sigma$ for all states in $U$.
We now collapse all states in $U$ into a single  \awesome state $u_a$ with an even priority, and a self-loop with payoff $3$, resulting in the SSG $\sys{G}_U$.
Now, if the maximizer can a.s.\ reach $U$ in $\sys{G}_1[\tau^*]$, then he can enforce $\MP{> 2} \cap \Parity$ in $\sys{G}_U[\tau^*]$.
All the remaining \awesome end-components in $\sys{G}_U$ satisfy $\MP{>0} \cap \Parity$ and so 
$\MP{>2} \cap \Parity$ due to \cref{lem:mean-payoff-2}.

We therefore fix a winning maximizer MD strategy $\sigma$ for $\MP{>2}$ and 
for each MD strategy $\tau$  write a linear program, consisting of
the gain-bias \emph{in}equations for gain of at least $2$ in $\sys{G}_U[\sigma,\tau]$,
and
forcing all biases to be non-negative and of polynomial size. This is a straight-forward adaptation of the gain-bias relations for solving mean-payoff MDPs (see, e.g., \cite[Theorem 8.2.6(a), p. 343]{P1994})
In particular, we have
\begin{align*}
b_{\tau,u} &< b_{\tau,\tau(u)} + \cost_1(u,\tau(u)) - 2 &\text{ for all } u \in \VMin \setminus U\\
b_{\tau,u} &< b_{\tau,\sigma(u)} + \cost_1(u,\sigma(u)) - 2 &\text{ for all } u \in \VMax \setminus U\\
b_{\tau,u} &< \sum_{(u,v) \in E} \prob(u,v) (b_{\tau,v} + \cost_1(u,v)) - 2 &\text{ for all } u \in \VRan \setminus U\\
b_{\tau,u} &\geq 0 &\text{ for all } u \in \{u_a\} \cup V \setminus U
\end{align*}
and we pick as the objective 
$$ \textit{Minimize} \sum_{u \in \{u_a\} \cup V \setminus U} b_{\tau,u}$$

It follows from the proof of Corollory 10.2a in \cite{schrijver1998theory} that the size of an optimal finite solution to such a linear program is at most $4m^2(m+1)(S+1)$, where $m$ is the number of variables and $S$ is the maximum size of any coefficient used. In our case $m \leq |V|+1$ and
$S \leq |\sys{G}|$, so the size of any $b_{\tau,u}$ in an optimal finite solution to such a linear program is of size polynomial in $|\sys{G}_1|$. Note that this loose upper bound, $B$, does not depend on $\tau$, $\sigma$ nor $U$.
    
We now build the SSG $\sys{G}_2=(V_2,E_2,\prob_2)$, where $\sys{G}_2\supseteq \sys{G}_U$, and the associated reward function $\cost_2$ that we derive from $\sys{G}_U$ by allowing maximizer to redistribute the rewards of random edges.
More precisely, let
$s$ be a random state with two outgoing edges $(s,t_1), (s,t_2) \in E$ and a unique predecessor $p\in \VMax{}$.
Then, for every MD minimizer strategy $\tau$, 
$\sys{G}_2$
 contain an extra random state $s_\tau$ and edges $(p,s_\tau),(s_\tau,t_1),(s_\tau,t_2)$---with the same probabilities $p_1$ and $p_2$ for taking $(s_\tau,t_1)$ and $(s_\tau,t_2)$ as for taking $(s,t_1)$ and $(s,t_2)$, respectively---and rewards
$\cost_2{(p,s_\tau)}\eqdef \cost_1{(p,s)}$,
$\cost_2{(s_\tau,t_1)} \eqdef \lfloor 1+b_{\tau,s}-b_{\tau,t_1} \rfloor$ and
$\cost_2{(s_\tau,t_2)} \eqdef \lfloor 1+b_{\tau,s}-b_{\tau,t_2} \rfloor$.
See \cref{fig:np.tradein} for an example.
Notice that, due to the inequalities defining the biases $b_{\tau,u}$, we have $p_1\cost_2{(s_\tau,t_1)} + p_2\cost_2{(s_\tau,t_1)} + 1 < p_1 \cost_1(s,t_1) + p_2 \cost_1(s,t_2)$, so maximizer sacrifices expected reward of at least $1$ at $s$.

\begin{figure}[t]
    \centering
    \includegraphics[width=0.8\linewidth]{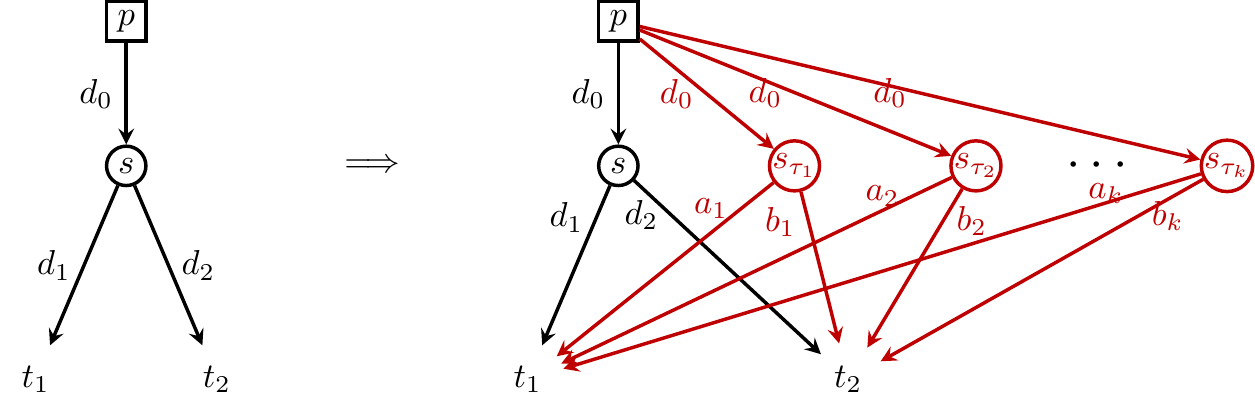}
    \caption{The reduction from $\sys{G}_1$ to $\sys{G}_2$, in which maximizer can choose to rebalance the rewards of edges out of probabilistic states at the cost of a reduced expected payoff where $a_i=\lfloor 1+b_{\tau_i,s}-b_{\tau_i,t_1} \rfloor$, $b_i=\lfloor 1+b_{\tau_i,s}-b_{\tau_i,t_2} \rfloor$ and $k$ is the number of MD minimizer strategies.}%
    \label{fig:np.tradein}
\end{figure}

This extended arena has the following property for every state $u\in V_2$ of $\sys{G}_2$.
\begin{lemma}
    \label{lem:reduction-gain-to-storage}
 Let $\tau$ be an MD minimizer strategy and $u$ a state. Then $u\in \AS[{\sys{G}_1[\tau]}]{\Gain}$ if, and only if, $u\in \bigcup_{k\ge 0}\AS[{\sys{G}_2[\tau]}]{\ES{k}\cap\Parity}$.
\end{lemma}
\begin{proof}

\emph{$(\impliedby)$}. %
Pick $k$ such that $u \in \AS[{\sys{G}_2[\tau]}]{\ES{k}\cap\Parity}$ holds, and  
let $\sigma$ be an a.s.\ winning FD strategy for the maximizer.
Now in $\sys{G}_1$, we simply follow $\sigma$, but 
whenever $\sigma$ picks a trade-in edge to $s_{\tau}$, we pick the original edge to $s$ instead. 
Notice that such a strategy ensures parity and the energy level at any point can only increase.
If such a strategy reaches a node in $U$ then it switches to an optimal strategy for
$\ES{k'}\cap\Parity$, where $k'$ is the minimum energy for which $\ES{k'}\cap\Parity$ holds for all states in $U$.
It is easy to see that while using such a strategy the energy can never drop more than $k+k'$, so it has to satisfy $\Gain$ a.s.

\emph{$(\implies)$}. %
First of all, note that due to the definition of biases 
$b_{\tau,u}$ we have that $\cost_2(u,u') > b_{\tau, u} - b_{\tau, u'} + 2$ for $u \in \VMax \cup \VMin$,
and $\cost_2(u,u') > b_{\tau, u} - b_{\tau, u'}$ for $u \in \{s_\tau | s \in \VRan\}$, because $\lfloor x + 1 \rfloor > x$ for all $x$.

Now pick any a.s.\ winning $\sigma$ for $\Gain$ in $\sys{G}_1[\tau]$.
Let $\sigma'$ be $\sigma$ that always picks trade-ins $s_\tau$ when possible.
Such a strategy still satisfies parity a.s.
Consider any play $\rho = s_0 e_0 s_1 e_1 s_2 e_2\ldots$ of $\sigma'$.
If $\rho$ reaches a state in $U$ then we switch at that point to an optimal strategy for $\ES{k'}\cap\Parity$ as defined above.
Otherwise, we have that for any infix $s_l e_l \ldots e_{h-1} s_h$ of $\rho$, the change in the energy level
is $\sum_{i=l}^{h-1} \cost_2(s_i, s_{i+1}) > \sum_{i=l}^{h-1} (b_{\tau, s_{i}} - b_{\tau, s_{i+1}}) = 
b_{\tau, s_{h}} - b_{\tau, s_{l}} \geq -B$.
This shows that $\ES{B+k'}\cap\Parity$ is satisfied a.s.\ by such a strategy.
\end{proof}

Using the existence of MD optimal minimizer strategies for their respective objectives in both games, we get the following.

\begin{restatable}{corollary}{NPGainSp}
    \label{lem:np:gain-sp}
$u\in \AS[{\sys{G}}]{\Gain} \iff u\in \bigcup_{k\ge 0}\AS[{\sys{G}_2}]{\ES{k}\cap\Parity}$.
\end{restatable}

%% file: app.6.NP.tradein.tex
\begin{proof}
    First of all, by the way $\sys{G}_1$ is defined, we have $u\in \AS[{\sys{G}}]{\Gain} \iff u\in \AS[{\sys{G}_1}]{\Gain}$.
    
    ($\Rightarrow$)
    For all MD strategies $\tau$ the following has to hold
    $u\in \AS[{\sys{G}_1[\tau]}]{\Gain}$.
    Due to Lemma \ref{lem:reduction-gain-to-storage} we get
    $u\in \bigcup_{k\ge 0}\AS[{\sys{G}_2[\tau]}]{\ES{k}\cap\Parity}$,
    so there exists $k$ such that $u \in \AS[{\sys{G}_2[\tau]}]{\ES{k}\cap\Parity}$.
    As there are only finitely many MD strategies, we let $k^*$ be the maximum value of $k$ corresponding to one of them.
    Note that $u \in \AS[{\sys{G}_2}]{\ES{k^*}\cap\Parity}$ has to hold, because
    $u \in \AS[{\sys{G}_2[\tau]}]{\ES{k^*}\cap\Parity}$ for all MD strategies $\tau$
    (as $\ES{k}\cap\Parity$ objective is upward-closed)
    and one of them has to be an optimal strategy for minimizer.
    
    ($\Leftarrow$) Suppose that 
    $u\not\in \AS[{\sys{G}_1}]{\Gain}$ then
    pick any MD optimal minimizer strategy $\tau$ such that 
    $u\not\in \AS[{\sys{G}_1[\tau]}]{\Gain}$.
    Due to Lemma \ref{lem:reduction-gain-to-storage} we get
    $u\not\in \bigcup_{k\ge 0}\AS[{\sys{G}_2[\tau]}]{\ES{k}\cap\Parity}$;
    a contradiction with the fact that $u\in \bigcup_{k\ge 0}\AS[{\sys{G}_2}]{\ES{k}\cap\Parity}$.
\end{proof}

\begin{figure}
    \begin{subfigure}[t]{0.5\textwidth}
        \centering
        \includegraphics[width=5cm]{Figures/np-example-G.pdf}
        \caption{The game $\sys{G}_1$}
        \label{fig:np2:a}
    \end{subfigure}\hfill
    \begin{subfigure}[t]{0.5\textwidth}
        \centering
        \includegraphics[width=6.5cm]{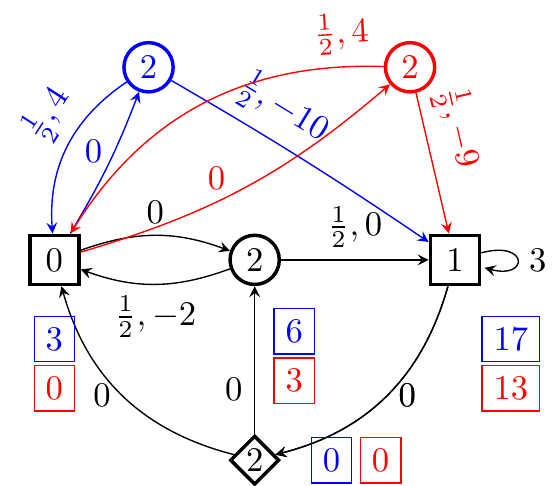}
        \caption{The derived game $\sys{G}_2$}
        \label{fig:np2:b}
    \end{subfigure}
    \caption{    \label{fig:np2}
        An example game $\sys{G}_1$ (left) and its example derived game $\sys{G}_2$.
        }
\end{figure}

\begin{example}[continuation of \cref{ex:gainNPExample}]
    \label{ex:gainNPExample2}
Consider the game $\sys{G}_1$ in \cref{fig:np2} (left).
In its derived game $\sys{G}_2$ there are as many trade-in options for the random state as there are MD minimizer's strategies (just two in this example).
The blue one (top left) corresponds to minimizer going left and the red one (top right) to going up.
Example biases that satisfy the inequalities presented in Section \ref{sec:gain.NP.tradein} are drawn next to the nodes inside colored boxes. They results in the rewards $4$ and $-10$ for the blue trade-in and $4$ and $-9$ for the red one.
\end{example}

%% file: sec.gain.NP.concise.tex
The final step is to show that we can clean up $\sys{G}_2$ by removing all but a small number of the new trade-in options for maximizer when entering a random state, preserving the fact that maximizer wins the $\ES{\cdot}\cap\Parity$ objective. Formally this whole subsection is dedicated to a proof of the following crucial lemma.

\begin{restatable}{lemma}{thmNPConcise}
    \label{lem:np:concise}
There exists a game $\sys{G}_3\supseteq\sys{G}_1$ that results from $\sys{G}_2$ by keeping, for any random state,
    at most twice the number of states in $\sys{G}_1$ trade-in options, and such that for any state $s\in V$
    maximizer wins the almost-sure $k$-storage-parity game in $\sys{G}_3$ iff he does in $\sys{G}_2$.
\end{restatable}

Most of the properties in this subsection hold for an arbitrary energy-parity game, so we will use $\sys{H}$ instead of $\sys{G}_2$ in order to avoid the use of double subscripts. 

The main idea of the proof of \cref{lem:np:concise} is to use the monotonicity of the 
$\ES{k}\cap\Parity$ objective with respect to the initial energy level $k$.
If maximizer a.s.~wins $\ES{\cdot}\cap\Parity$ from state $p$
then there is a least $k_p\in\N$ such that (for some $l$),
$\ES{k_p,l}\cap\Parity$ holds a.s.
Fix $l$ large enough to work for all minimal $k_p$ for every state $p$---and for all purposes of the proofs below.

Consider a configuration $(p,k_p)\in\AS{\ES{k_p,l}\cap\Parity}$
where $p$ has newly introduced outgoing edges that allow for trade-ins (it has a random successor node).
Let $\sigma$ be a winning maximiser strategy for this game that depends only on the state and the energy level in the energy store%
\footnote{Recall that such a strategy must exist as, once the store limit $l$ is fixed, the game becomes an ordinary finite parity game.}, and
let $\sigma_{\min}$ denote the maximiser strategy that maps each maximiser state $p$ to the successor that $\sigma$ assigns to $(p,k_p)$.
Note that this strategy is positional, and therefore uses only one possible trade-in option.

We first observe that maximiser can ensure by using this strategy that he can only gain energy distance relative to the minimal energy level of the state (except where the energy is limited by the capacity of his energy store):
For every run $(s_0,k_0),(s_1,k_1),(s_2,k_s),(s_3,k_3),\ldots$
of $\sys{H}$ consistent with $\sigma_{\min}$ and all $i \in \omega$ it holds that $k_{i+1} - k_{s_{i+1}} \geq k_i - k_{s_i}$.
The following lemma is a direct consequence.

\begin{lemma}
    The strategy $\sigma_{\min}$ almost-surely guarantees that 1) the cumulative rewards tend to infinity or 2) the parity condition holds.
    That is, for every minimizer strategy $\tau$ and initial state $s$ of $\sys{H}$ it holds that $\Prob[\sys{H},\sigma_{\min},\tau]{s}(\LimInf{=\infty}\cup\Parity) = 1$.
\end{lemma}
\begin{proof}
    Assume for contradiction that minimizer has a strategy that ensures that runs with a positive probability weight contain (1) only finitely many transitions that lead to a true gain in energy (relative to the minimal energy level) and (2) do not satisfy the parity condition.
    (1) is a co-B\"uchi objective, (2) a parity objective, so (1) and (2) together are a parity objective. Thus, minimizer has a memoryless strategy $\tau$ to obtain this.
    Thus $\sys{H}[\sigma_{\min},\tau]$ has a leaf-component where this holds.
    Thus, $\sys{H}[\sigma,\tau]$ is not winning on the states of this leaf-component on the minimal energy level. (contradiction)
\end{proof}

We call the property $(\LimInf{=\infty}\cup\Parity)$ established by this lemma the \emph{lift or win property}
and will use it for a separation of concerns.
For this, we first show that, when the dominating priority is odd, then the maximizer can win on a smaller set that he can ensure is never left while winning the energy storage condition almost surely.

For a set $S$ of states, we write 
$\atr_i^{\sys{H}}(S)$ 
for the set of states from which player $i\in \{\Max,\Min\}$ (maximizer / minimizer) can force the game to a state in $S$.
In particular, $\atr_\Max^\sys{H}(S) = \AS{\eventually S}$ is the set of states for which maximizer can ensure to almost-surely reach $S$.
We call a set $S$ of states a (minimizer) \emph{trap} if all minimizer states and all random states in $S$ have only successors in $S$. Naturally, the union of two traps is also a trap, so there exists a unique $\subseteq$-maximal trap.

\begin{lemma}
    \label{lem:winningTrap}
    Let $\sys{H}$ be a game with minimal odd priority $o$, where the maximiser wins storage parity from all positions, and let $S_o$ be the states with priority $o$. 
    Then there is a trap $S_t$ in $\sys{H}\setminus \atr_\Min^{\sys{H}}(S_o)$, such that the maximiser wins storage parity from all positions in the subgame $\sys{H}\cap S_t$, that is, without exiting $S_t$.
\end{lemma}
\begin{proof}
    Assume for contradiction that no such trap exists.
    Then the minimiser has an almost-sure winning strategy---and thus a positional winning strategy $\tau$---for all positions in $\sys{H}\setminus \atr_\Min^{\sys{H}}(S_o)$.
    Then the minimiser can win almost surely in $\sys{H}$ by a positional winning strategy that fixes an arbitrary strategy for her positions in $S_o$, uses her attractor strategy in all her other positions in $ \atr_\Min^{\sys{H}}(S_o)$, and $\tau$ elsewhere. (contradiction)
\end{proof}

The minimal energy level for winning from a state in $S_t$ can, of course, differ from the minimal sufficient energy level for the same state in the full game $\sys{H}$. We now partition the winning regions using divide and conquer.

\begin{lemma}
    \label{lem:NP:DnC}
    Let $\sys{H}$ be a game where the maximiser wins storage parity from all positions.
    Let $o$ be the minimal odd priority that occurs in $\sys{H}$.
    If $o$ is the minimal priority in $\sys{H}$ then let $S_t$ be defined as the trap $S_t$ guaranteed by \cref{lem:winningTrap},
    otherwise let $S_t$ be the set of states with smaller priority than $o$. The following holds.
    \begin{enumerate}
        \item Maximizer wins storage parity from all positions in
        the subgame $\sys{H}' = \sys{H} \setminus \atr_\Max^{\sys{H}}(S_t)$.
        \item Fix maximizer strategies $\sigma_1, \sigma_2$ and $\sigma_3$
        that are almost-sure winning for 1) storage-parity in $\sys{H}'$, storage-parity in $S_t$, and 3) reachability ($\eventually{S_t}$), respectively, and let $\sys{I}\subseteq \sys{H}$ be the game in which all new trade-in states that are never used by those strategies are removed.
        Maximizer almost-surely wins the storage-parity objective from all states of $\sys{I}$.
    \end{enumerate}
\end{lemma}

\begin{proof}
    For the first part, one immediately sees that a winning minimizer strategy for (some) states in $\sys{H}'$ would also be winning for these states in $\sys{H}$.
    
    For the second part, notice that maximizer can combine the existing strategies into an overall winning strategy as follows.
    
    Suppose $k_t\in\N$ is large enough so that for all states in $S_t$, maximizer wins the storage-parity objective $\ES{k_t}\cap\Parity$.
    Based on this, we can pick $k_a\in\N$ large enough so that,
    in a game that starts in $\atr_{\Max}^{\sys{H}}(S_t)$
    with energy $k_a$ and where maximiser plays the attractor strategy towards $S_t$, he has a positive chance of reaching $S_t$ with energy $\ge k_t$ while remaining in the almost-sure winning region for storage-parity (in $\sys{H}$).
    Finally, let $k>\max\{k_a,k_t\}$ be large enough so that for all states in $\sys{H'}$, maximizer wins the storage-parity objective $\ES{k}\cap\Parity$. W.l.o.g., this is already witnessed by the strategy $\sigma_1$, by monotonicity of the objective.
    Maximizer will play as follows.
    
    As long as the energy level is low ($< k$),
    maximizer plays according to $\sigma_{\min}$.
    By the \emph{lift or win} property (\cref{lem:mean-payoff-2}) he can either win or gain an arbitrary amount of energy.
    Alternatively, assuming he is in $\atr_\Max^{\sys{H}}(S_t)$ and has sufficient energy, he invests it into an attempt to reach $S_t$ in $\atr_\Max^{\sys{H}}(S_t)$ while complying with the minimal energy level on the way and, if $o$ is the minimal priority, having sufficient energy in $S_t$ to win storage parity in the trap $S_t$.
    Outside of $\atr_\Max^{\sys{H}}(S_t)$, 
    he plays according to $\sigma_1$, the winning strategy in $\sys{H'}$, while maintaining an energy level of at least $k$.
    This combined strategy is winning for the $k$-storage-parity because 1) it remains in the almost-sure winning region
    and 2) either eventually forever follows a winning strategy in $S_t\cup \sys{H}'$, or 
    (in case the minimal priority is even) infinitely often tries to reach states with the dominant priority.
    
    As this strategy only combines the existing strategies, it will never use any trade-in state in $\sys{H}\setminus\sys{I}$,
    and therefore works in the smaller subgame $\sys{I}$.
\end{proof}

This finally allows us to establish our main claim, of which
\cref{lem:np:concise} is a direct consequence.

\begin{lemma}
    If the maximiser almost surely wins storage parity for $\sys{G}_2$, he can win storage parity in $\sys{G}_2$ with a strategy that does not use more choices for any state in $\sys{G}_2$ than twice the number of states of $\sys{G}_1$ has states.
\end{lemma}

\begin{proof}
    The claim follows from a recursive application of \cref{lem:NP:DnC}.
    Starting with $\sys{H}\subseteq \sys{G}_2$ defined by the almost-sure winning states, each application will split the game into disjoint subgames $\sys{H}'$, $S_t$, and $\atr_\Max^{\sys{H}}(S_t)\setminus S_t$, in which maximizer can be assumed to win according to simpler (wrt.~the number of trade-ins used) strategies.
    Notice that every new trade-in states $s_\tau \in \sys{G}_2 \setminus \sys{G}_1$ will belong to the same subgame as its accompanying original random state $s\in \sys{G}_1$.
    In every decomposition $S_t$ must be non-empty, so the number of states in $\sys{G}_1$ bounds the recursion depth.
    
    The base cases are either empty or games in which maximizer wins only by combining $\sigma_{\min}$ and an attractor strategy towards the dominating priority. Both can be chosen MD.
    In any further decomposition, any given state will wither belong to a smaller game ($\sys{H}'$ or $S_t$), in which case the number of necessary trade-in options is unchanged,
    or is in $\atr_\Max^{\sys{H}}(S_t)\setminus S_t$, in which case
    the combined strategy may need to chose between $\sigma_{\min}$
    and an attractor strategy. But notice that the choice of trade-in state is meaningless for the attractor strategy, because all such states have the same (distributions over) successors.
\end{proof}

\begin{figure}
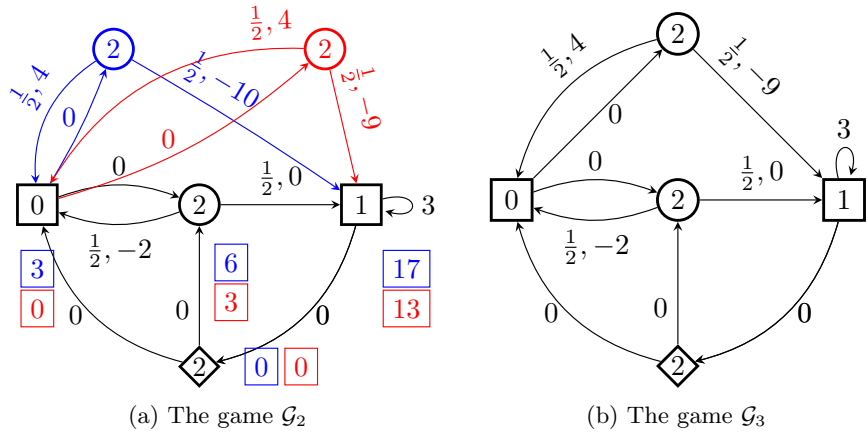

    \begin{subfigure}[t]{0.5\textwidth}
        \centering
        \includegraphics[width=6cm]{Figures/np-example-G2.pdf}
        \caption{The game $\sys{G}_2$}
        \label{fig:np3:b}
    \end{subfigure}\hfill
    \begin{subfigure}[t]{0.5\textwidth}
        \centering
        \includegraphics[width=5cm]{Figures/np-example-G3.pdf}
        \caption{The game $\sys{G}_3$}
        \label{fig:np3:c}
    \end{subfigure}
    \caption{    \label{fig:np3}
        An example game $\sys{G}_2$ (left) and the derived games.
    }
  \end{figure}

\begin{example}[continuation of \cref{ex:gainNPExample2}]
    \label{ex:g1-3}
    Consider the game $\sys{G}_2$ in \cref{fig:np3} (left). 
    We can prune $\sys{G}_2$ into a game where all but one new alternative state is removed.
    In this game $\sys{G}_3$, depicted on the right, maximizer can almost-surely guarantee the $\Gain$ condition while simultaneously ensuring that no negative cycle is closed. This means that $\ES{k}\cap\Parity$ holds almost-surely
    in $\sys{G}_3$, and hence $\EN{k}\cap\Parity$ in $\sys{G}$.
\end{example}

\subsubsection{Proof of \cref{thm:as-gain-np}}
\rule[-10pt]{0pt}{10pt}\\

We are now ready to prove the main theorem of \cref{sec:gain}.

\thmASGainNP*
\begin{proof}
Guess a game $\sys{G}_3$ that uses only the given bound on the number of choices, i.e., without constructing the exponentially large game $\sys{G}_2$. Prune the unreachable random states 
and verify that maximizer can almost-surely ensure the storage-parity objective in $\sys{G}_3$.
The correctness of this procedure follows from
\cref{lem:np:concise,lem:np:gain-sp}.
\end{proof}